\documentclass[12pt,onecolumn]{IEEEtran}

\usepackage{cite}
\usepackage{pslatex,bm}
\usepackage{epsfig}
\usepackage{amsmath}
\usepackage{amssymb}
\usepackage{array}
\usepackage{multicol}
\usepackage{color}
\usepackage{url}
\usepackage{dsfont}
\psfull

\newtheorem{Theo}{Theorem}
\newtheorem{lemma}{Lemma}

\newtheorem{definition}{Definition}

\newcommand{\vectornorm}[1]{\left|\left|#1\right|\right|}

\DeclareMathAlphabet{\mathpzc}{OT1}{pzc}{m}{it}

\renewcommand{\baselinestretch}{1.9}

\def\calW{{\mathcal W}}

\def\bv{{\mathbf v}}

\def\bS{{\mathbf S}}

\def\tr{{\rm tr}}

\def\E{{\mathbf E}}

\def\b0{{\mathbf 0}}

\begin{document}

\title{On Secrecy Rate of the Generalized Artificial-Noise Assisted Secure Beamforming for Wiretap Channels}

\author{
\authorblockN{ \vspace{-4mm} Pin-Hsun Lin \hspace{0.5cm}
 Szu-Hsiang Lai \hspace{0.5cm} Shih-Chun Lin \hspace{0.5cm} Hsuan-Jung Su}

\thanks{
Pin-Hsun Lin, and Hsuan-Jung Su are with Department of Electrical
Engineering and Graduate Institute of Communication Engineering,
National Taiwan University, Taipei, Taiwan 10617. Szu-Hsiang Lai is with MStar semiconductor, Hsinchu, Taiwan 114. Shih-Chun Lin
is with Graduate Institute of Computer and Communication Engineering,
National Taipei University of Technology, Taipei, Taiwan 10643. Emails: \{r98942061@ntu.edu.tw, pinhsunlin@gmail.com, sclin@ntut.edu.tw,
hjsu@cc.ee.ntu.edu.tw\}. The material in
this paper was presented in part at The 22nd IEEE Symposium on Personal, Indoor, Mobile and Radio Communications (PIMRC 2011).
 This work was supported by
the National Science Council, Taiwan, R.O.C., under grant NSC 99-2628-E-002-001 and 100-2221-E-002-133. } }
\maketitle \thispagestyle{empty} \vspace{-15mm}
{\renewcommand{\baselinestretch}{2}
\begin{abstract}
In this paper we consider the secure transmission in fast Rayleigh fading channels
with full knowledge of the main channel and only the statistics of the eavesdropper's channel state
information at the transmitter. For the multiple-input, single-output, single-antenna eavesdropper
systems, we generalize Goel and Negi's celebrated
artificial-noise (AN) assisted beamforming, which just selects the directions to transmit AN heuristically.
Our scheme may inject AN to the direction of the message, which outperforms Goel and Negi's scheme
where AN is only injected in the directions orthogonal to the main channel. The ergodic secrecy rate of the
proposed AN scheme can be represented by a highly simplified power allocation problem. To attain
it, we prove that the optimal transmission scheme for the message bearing signal is a beamformer,
which is aligned to the direction of the legitimate channel. After characterizing the optimal eigenvectors
of the covariance matrices of signal and AN, we also provide the necessary condition for transmitting
AN in the main channel to be optimal. Since the resulting secrecy rate is a non-convex power allocation
problem, we develop an algorithm to efficiently solve it. Simulation
results show that our generalized AN scheme outperforms Goel and Negi's,
especially when the quality of legitimate channel is much worse than that of eavesdropper's. In particular, the
regime with non-zero secrecy rate is enlarged, which can
significantly improve the connectivity of the secure network when the proposed AN
assisted beamforming is applied.
\end{abstract}

\section{ Introduction}
In a wiretap channel, a source node wishes to transmit
confidential messages securely to a legitimate receiver and to
keep the eavesdropper as ignorant of the message as possible. As a
special case of the broadcast channels with confidential messages
\cite{csiszar1978broadcast}, Wyner \cite{Wyner_wiretap}
characterized the secrecy capacity of the discrete memoryless
wiretap channel. The secrecy capacity is the largest rate
communicated between the source and destination nodes with the
eavesdropper knowing no information of the messages. Motivated by
the demand of high data rate transmission and improving the
connectivity of the network \cite{Secureconnect}, the multiple
antenna systems with security concern are considered by several
authors. With full channel state information at the transmitter
(CSIT), Shafiee and Ulukus \cite{Shafiee_secrecy_2_2_1_J} first
proved the secrecy capacity of a Gaussian channel with two-input,
two-output, single-antenna-eavesdropper. Then the authors of
\cite{Khisti_MIMOME,Oggier_MIMOME, Liu_MIMO_wiretap} extended the
secrecy capacity to the Gaussian multiple-input multiple-output
(MIMO), multiple-antenna-eavesdropper channel using different
techniques. On the other hand, due to the characteristics of
wireless channels, the impacts of fading channels on the secrecy
transmission were considered in \cite{Liang_fading_secrecy,
Khisti_MIMOME} with full CSIT. Considering practical issues such
as the limited bandwidth of the feedback channels or the speed of
the channel estimation at the receiver, the perfect CSIT may not
be available. Therefore, several works considered the secrecy
transmission with partial CSIT
\cite{Goel_AN,gopala2008secrecy,Khisti_MISOME,Li_fading_secrecy_j,Pulu_ergodic}.
In \cite{Goel_AN,gopala2008secrecy,Khisti_MISOME}, the authors
naively chose the directions of signal and AN without optimization
and the resulting performance is suboptimal. In addition, they
solved the power allocation via full search, which is inefficient.
Furthermore, they did not prove the equality of the power
constraint is hold (using all power is optimal). In
\cite{Li_fading_secrecy_j}, a single antenna system is considered,
thus the authors did not solve the beamformer and power allocation
problems. Also, the authors did not prove the rate increases with
increasing total power. In \cite{Pulu_ergodic}, the authors did
not consider the AN in the transmission, and thus their scheme is
a special case of ours. Indeed, as shown in
\cite{Goel_AN,Khisti_MISOME}, adding AN in transmission is crucial
in increasing the secrecy rate in fading wiretap channels. Also
under the case that the main channel is fully known at
transmitter, the optimal direction for signals is not solved
analytically in \cite{Pulu_ergodic}. However, the secrecy
capacities for channels with partial CSIT are known only for some
limited cases, i.e., the transmitter has single antenna with block
fading \cite{gopala2008secrecy} and only the statistics of both
the main and eavesdropper's channels are known at the transmitter
\cite{Lin_Ergodic_secrecy_capacity}.

In this paper, we consider an important type of wiretap channels
with partial CSIT, namely, the multiple-input single-output
single-antenna-eavesdropper (MISOSE) fading wiretap
channels. We assume that the main channel has a constant
channel gain and the eavesdropper channel is fast faded,
respectively. We also assume that the transmitter has perfect knowledge
of the main channel and only the statistics of the eavesdropper
channel. We adopt the artificial noise (AN) assisted secure
beamforming as our transmission scheme, where the AN is used to
disrupt the eavesdropper's reception
\cite{Khisti_MISOME}\cite{Goel_AN}. Although the secrecy
capacity of the considered channel is unknown, the performance of
the AN-assisted beamforming has been shown to be
capacity-achieving in the high signal to noise ratio (SNR) regime when the transmitter is
equipped with a large number of antennas \cite{Khisti_MISOME}.
However, in other operation regimes, the heuristically selected
directions in \cite{Khisti_MISOME}\cite{Goel_AN} to transmit AN
may not be optimal, where the AN is restricted to be in the null
space of the legitimate channel. This motivates our study on
optimizing the AN assisted secure beamforming.
Note that the assumption that the statistics of the eavesdropper's channel are known
at transmitter was also used in \cite{Goel_AN} to design the power allocation between
the signal and the AN (see \cite[(8)]{Goel_AN}). Thus our comparison to the method in
\cite{Goel_AN} in Section \ref{Sec_simulation} is reasonable and fair.

The main contribution of our paper is that we propose a general AN scheme, which outperforms \cite{Goel_AN}. More specifically, the optimal AN may be full rank under some channel conditions rather than low rank, as restricted in \cite{Goel_AN}. In addition, we provide a simplified power allocation problem to describe the ergodic secrecy rate, which highly reduces the complexity of solving the rate. To attain it, we characterize the optimal beamforming directions and the power allocation strategies for AN. We also provide the necessary
condition for transmitting AN in the main channel to be optimal. After characterizing the eigenvectors of the covariance matrices of signal and AN, the resulting rate becomes a non-convex power allocation problem and we develop an algorithm to efficiently solve it. Simulation results confirm
that the full-rank AN provides rate gains over \cite{Goel_AN}, especially through the enlarged non-zero rate region.
Note that the secure connectivity in a network is assured by the
non-zero secrecy rate of the transmitter-receiver pairs
\cite{Secureconnect}. Thus our scheme is very useful for the
large scale wireless network applications, which is an important type of
applications of the MISOSE wiretap channels \cite{Secureconnect}.

The rest of the paper is organized as follows. In Section
\ref{Sec_system_model} we introduce the considered system model. In
Section \ref{Sec_general_P_optimal_input_cov_LB} an intuitive explanation of
the rate gain from the proposed scheme is provided. We then develop our main result, i.e., the ergodic secrecy rate, via three steps. In this section we also provide the necessary condition to have a full rank optimal covariance matrix of AN. In Section \ref{Sec_proposed algorithm}, we provide an iterative algorithm to solve the power allocation problem. In Section \ref{Sec_simulation} we
demonstrate the simulation results. Finally, Section
\ref{Sec_conclusion} concludes this paper.

\section{System model}\label{Sec_system_model}
In this paper, lower and upper case bold alphabets denote vectors and matrices, respectively. The superscript $(.)^H$ denotes the transpose complex
conjugate. $|\mathbf{A}|$ and $|a|$ represent the determinant of the
square matrix $\mathbf{A}$ and the absolute value of the scalar
variable $a$, respectively. A diagonal matrix whose diagonal
entries are $a_1 \ldots a_k$ is denoted by
$diag(a_1 \ldots a_k)$. The trace of $\mathbf{A}$
is denoted by $\tr(\mathbf{A})$. We define $C(x) \triangleq
\log(1+x)$ and $(x)^+\triangleq\max\{0,\,x\}$. $\mathbf{A}^{\perp}$ is the null space of $\mathbf{A}$. The mutual
information between two random variables is denoted by $I(;)$. $\mathbf{I}_n$ denotes the
$n$ by $n$ identity matrix. $\mathbf{A}\succ 0$ and $\mathbf{A}\succeq 0$ denote
that $\mathbf{A}$ is a positive definite and positive semi-definite matrix, respectively. $\mathbf{a}\succ \mathbf{b}$ denotes $\mathbf{a}$ majorizes $\mathbf{b}$.

We consider the MISOSE system as shown in Fig. \ref{Fig_system},
where the transmitter (Alice) has $n_T$ antennas and the legitimate receiver (Bob) and the eavesdropper (Eve) each has single antenna. The received
signals at Bob and Eve can be respectively represented as
\begin{align}
y_k&=\mathbf{h}^H\mathbf{x}_k+n_{1,k}, \label{EQ_main}\\
z_k&=\mathbf{g}^H_k\mathbf{x}_k+n_{2,k}\label{EQ_Eve},
\end{align}
where $\mathbf{x}_k\in\mathds{C}^{n_T\times 1}$ is the transmit vector, $k$ is the time
index, $\mathbf{h}$ is the constant main channel vector,
$\mathbf{g}_k\sim CN(0, \mathbf{I}_{n_T})$ is the random eavesdropper's channel, and $n_{1,k}$ and
$n_{2,k}$ are circularly symmetric complex additive white Gaussian
noises with variances one at Bob and Eve, respectively. In this system model, we assume that full
CSI of the legitimate channel and only the statistics of Eve's
channel are known at transmitter. Without loss of generality, in the following we omit the time index to simplify the notation.

The perfect secrecy and secrecy capacity are defined as follows.
Consider a $(2^{nR}, n)$-code with an encoder that maps the
message $w\in \calW=\{1,2,\ldots, 2^{nR}\}$ into a length-$n$
codeword, and a decoder at the legitimate receiver that maps the
received sequence $y^n$ (the collections of $y$ over code length
$n$) from the MISOSE channels \eqref{EQ_main} to an estimated
message $\hat w\in\calW$. We then have the following definition of secrecy capacity.
\begin{definition}[Secrecy Capacity
\cite{gopala2008secrecy}] \label{Def_Perfect} {\it Perfect secrecy
is achievable with rate $R$ if, for any positive $\varepsilon$ and $\varepsilon'$, there
exists a sequence of $(2^{nR}, n)$-codes and an integer $n_0$ such
that for any $n>n_0$
\begin{equation}\label{eq_equivocation_given_h}
I(w;z^n,\mathbf{h}^n, \mathbf{g}^n)/n<\varepsilon, \mathrm{and
}\;\;{\rm Pr}(\hat{w}\neq w) \leq \varepsilon',
\end{equation}
where $w$ is the secret message, $z^n$, $\mathbf{h}^n$, and
$\mathbf{g}^n$ are the collections of $z, \,\,\mathbf{h}$, and
$\mathbf{g}$ over code length $n$, respectively. The {\rm secrecy capacity} $C_s$ is the supremum of all achievable
secrecy rates.}
\end{definition}

From Csisz$\acute{\mbox{a}}$r and
K$\ddot{\mbox{o}}$rner's argument \cite{csiszar1978broadcast}, we
know that the general secrecy capacity can be represented by
\begin{align}
C=\underset{p(\mathbf{x}|\mathbf{u}),\,p(\mathbf{u}) }{\max}I(\mathbf{u};y)-I(\mathbf{u};z|\mathbf{g}).\label{EQ_partial_CSI_rate}
\end{align}

However, for our considered CSIT setting, which is not full CSIT, the optimal $p(\mathbf{x}|\mathbf{u})$ and $p(\mathbf{u})$ are still unknown. We propose to apply the linear channel prefixing and Gaussian signaling to $f(x|u)$ as
\begin{equation}\label{EQ_prefixing}
\mathbf{x}=\mathbf{u}+\mathbf{v},
\end{equation}
where $\mathbf{u}\sim CN(0,\mathbf{S_u})$ and $\mathbf{v}\sim
CN(0,\mathbf{S_v})$ are independent vectors to convey the message
and AN, respectively. In addition, the feasible channel input
matrices of signal and AN belong to the set
\begin{align}
 \textsl{S}=\{(\mathbf{S_u},\mathbf{S_v}):\mbox{tr}&(\mathbf{S_u}+\mathbf{S_v})\leq P_T, \mathbf{S_u} \succeq 0, \mathbf{S_v} \succeq
 0\}.
\label{EQ_power_constraint}
\end{align}

Substituting \eqref{EQ_main}, \eqref{EQ_Eve}, and
\eqref{EQ_prefixing} into \eqref{EQ_partial_CSI_rate}, we have the ergodic
secrecy rate with generalized AN (GAN) as
\begin{align}\label{EQ_general_rate_AN}
R_{GAN}=\max_{ \mathbf{S_u},\,\mathbf{S_v}\in
\textsl{S}}\left(\log
\left(\frac{1+\mathbf{h}^H\left(\mathbf{S_u}+\mathbf{S_v}\right)\mathbf{h}}{1+\mathbf{h}^H\mathbf{S_v}\mathbf{h}}\right)
-\E\left[\log
\left(\frac{1+\mathbf{g}^H\left(\mathbf{S_u}+\mathbf{S_v}\right)\mathbf{g}}{1+\mathbf{g}^H\mathbf{S_v}\mathbf{g}}\right)\right]\right)^+.
\end{align}

Note that we do not limit the covariance matrix $\mathbf{S_v}$ of
the AN $\bv$ to have any special structure besides the
conventional one \eqref{EQ_power_constraint}. Thus our GAN scheme
generalizes the AN in \cite{Goel_AN}, which is \textit{only}
allowed to be transmitted in the null space of the main channel.
On the contrary, our GAN can be transmitted in all possible
directions. We then solve the ergodic secrecy rate optimization
problem \eqref{EQ_general_rate_AN} for the proposed GAN
beamforming (GAN-BF) scheme in the following sections.

\section{Optimization of the ergodic secrecy rate}\label{Sec_general_P_optimal_input_cov_LB}
In this section, we identify the structure of the optimal
solutions $\mathbf{S^*_u}$ and $\mathbf{S^*_v}$ for the GAN-BF
optimization problem \eqref{EQ_general_rate_AN}, where AN is not
restricted in the null space of the main channel. By exploiting
the optimal structure, we transform the complicated optimization
problem over the covariance matrices \eqref{EQ_general_rate_AN} as
a much simpler one in Theorem \ref{Th_main_result}. In the
following Theorem \ref{Th_main_result}, the optimized ergodic
secrecy rate of the GAN-BF is merely characterized by the power
allocations among the message bearing signal, AN in the direction
of the main channel, and AN in the directions orthogonal to the
main channel.

\begin{Theo}\label{Th_main_result}
For the MISOSE fast fading wiretap channel with the perfect
information of the legitimate channel $\mathbf{h}$, and only the
statistics of the eavesdropper's channel
$\mathbf{g}\sim CN(0, \mathbf{I}_{n_T})$ known at the transmitter, the optimization of the secrecy rate in \eqref{EQ_general_rate_AN} can be reduced to the following optimization problem
\begin{align}\label{EQ_simplified_obj}
R_{GAN}=\underset{\overset{P_{U},\,P_{V_1},\,P_{V_2}:}{P_U+P_{V_1}+(n_T-1)P_{V_2}{ =} P_T}}\max\left(\log\left(1+\frac{||\mathbf{h}||^2P_U}{1+||\mathbf{h}||^2P_{V_1}}\right)-\E\left[\log\left(1+\frac{\tilde{G}_1P_U}{1+\tilde{G}_1P_{V_1}+\left(\underset{i=2}{\overset{n_T}\sum}\tilde{G}_i \right)P_{V_2}}\right)\right]\right)^+,
\end{align}
where $P_U,\,P_{V_1},\mbox{ and, }P_{V_2}$ are the powers of the signal, the AN in the main channel, and the AN in the null space of the main channel, respectively. $\tilde{G}_i\triangleq|g_i|^2\sim EXP(1)$, which is the exponential distribution with mean equal to 1, for
$i=1,2,\ldots,n_T$.
\end{Theo}
Comparing \eqref{EQ_general_rate_AN} to \eqref{EQ_simplified_obj} we
can easily find that the optimization problem is vastly simplified
from solving two matrices to three scalar variables. Note that we
divide the proof of Theorem \ref{Th_main_result} into three parts
for the tractability and each part corresponds to Theorem 2, Lemma
3, and Lemma 4, respectively. Before proving
(\ref{EQ_simplified_obj}), we introduce two important lemmas to
proceed.

\begin{lemma}\label{lemma_optimal_direction}
Given a diagonal matrix $\mathbf{D}=diag(d_1,\,d_2,\,\cdots,\,d_n)\in\mathds{C}^{n\times n}$. Assume $d_1\geq d_2 \geq \cdots \geq d_n$ and $\mathbf{U}$ is unitary. Then $\mathbf{U}=[\mathbf{h}/||\mathbf{h}||,\,\mathbf{h}^{\perp}/||\mathbf{h}||]$ and $\mathbf{U}=[\mathbf{h}^{\perp}/||\mathbf{h}||,\,\mathbf{h}/||\mathbf{h}||]$ maximizes and minimizes $\mathbf{h}^H\mathbf{U}\mathbf{D}\mathbf{U}^H\mathbf{h}$, respectively.
\end{lemma}
\begin{proof}
We can rewrite the maximization problem in the statement of the lemma as
\begin{align}\label{EQ_Lemma1_proof}
\max\,\sum_{i=1}^{n} \,d_i|\tilde{h}_i|^2,\,s.t.\,\sum_{i=1}^{n}\,|\tilde{h}_i|^2=||\mathbf{h}||^2,
\end{align}
where $\tilde{\mathbf{h}}=\mathbf{U}^H\mathbf{h}$, $\tilde{h}_i$ is the $i$th entry of $\tilde{\mathbf{h}}$.
Then it can be easily seen that $|\tilde{h}_1|=||\mathbf{h}||$ with $|\tilde{h}_2|=|\tilde{h}_3|=\cdots=|\tilde{h}_n|=0$
can optimize \eqref{EQ_Lemma1_proof}. Therefore, it is clear that $\mathbf{U}=[\mathbf{h}/||\mathbf{h}||,\,\mathbf{h}^{\perp}/||\mathbf{h}||]$.
The minimization part can be proved similarly.
\end{proof}

Now, we identify the eigenvectors of the optimal $\mathbf{S_u^*}$
and $\mathbf{S_v^*}$ through the following lemma.

\begin{lemma}\label{Th_U_eq_V}
The optimal covariance matrices of the signal and AN
$\mathbf{S_u^*}$ and $\mathbf{S_v^*}$ for
\eqref{EQ_general_rate_AN} have the same eigenvectors as
$\mathbf{[h/||h||,\,h^{\perp}/||h||]}$.
\end{lemma}
\begin{proof}
Assume $\mathbf{S_u+S_v}$ and $\mathbf{S_v}$ are eigen-decomposed as $\mathbf{U}\mathbf{D}_1\mathbf{U}^H$ and $\mathbf{V}\mathbf{D}_2\mathbf{V}^H$, respectively. First, we can reform \eqref{EQ_general_rate_AN2} as
\begin{align}
\underset{\mathbf{S_u},\,\mathbf{S_v}}\max\,R=\underset{\mathbf{D}_1,\,\mathbf{D}_2}\max\underset{\mathbf{U},\,\mathbf{V}}\max\,R
=\underset{\mathbf{D}_1,\,\mathbf{D}_2}\max\underset{\mathbf{U},\,\mathbf{V}}\max\,\left(\log
\left(\frac{1+\mathbf{h}^H\mathbf{U}\mathbf{D}_1\mathbf{U}^H\mathbf{h}}{1+\mathbf{h}^H\mathbf{V}\mathbf{D}_2\mathbf{V}^H\mathbf{h}}\right)
-\E\left[\log
\left(\frac{1+\mathbf{g}^H\mathbf{U}\mathbf{D}_1\mathbf{U}^H\mathbf{g}}{1+\mathbf{g}^H\mathbf{V}\mathbf{D}_2\mathbf{V}^H\mathbf{g}}\right)\right]\right)^+.\label{EQ_origin_Th2_problem}
\end{align}
Since $\mathbf{g}$ is isotropically distributed,
\[\E\left[\log
\left(\frac{1+\mathbf{g}^H\mathbf{U}\mathbf{D}_1\mathbf{U}^H\mathbf{g}}{1+\mathbf{g}^H\mathbf{V}\mathbf{D}_2\mathbf{V}^H\mathbf{g}}\right)\right]=\E\left[\log
\left(\frac{1+\mathbf{g}^H\mathbf{D}_1\mathbf{g}}{1+\mathbf{g}^H\mathbf{D}_2\mathbf{g}}\right)\right],
\]
which is independent of $\mathbf{U}$ and $\mathbf{V}$. Thus the inner optimization problem on the right hand side (RHS) of (\ref{EQ_origin_Th2_problem}) becomes
\begin{equation}
(\mathbf{U}^*,\,\mathbf{ V}^*)=\arg\underset{\mathbf{U},\,\mathbf{V}}\max\,\log
\left(\frac{1+\mathbf{h}^H\mathbf{U}\mathbf{D}_1\mathbf{U}^H\mathbf{h}}{1+\mathbf{h}^H\mathbf{V}\mathbf{D}_2\mathbf{V}^H\mathbf{h}}\right).
\end{equation}
Then from Lemma \ref{lemma_optimal_direction} we know that $\mathbf{U}=\Pi_{\mathbf{U}}[\mathbf{h}/||\mathbf{h}||,\,\mathbf{h}^{\perp}/||\mathbf{h}||]$
and $\mathbf{V}=\Pi_{\mathbf{V}}[\mathbf{h}/||\mathbf{h}||,\,\mathbf{h}^{\perp}/||\mathbf{h}||]$ can simultaneously maximize
 and minimize the numerator and denominator, respectively, where $\Pi_{\mathbf{U}}$ and $\Pi_{\mathbf{V}}$ are the permutation matrices such that the eigenvector
 $\mathbf{h}/||\mathbf{h}||$ is in the direction of the maximum and minimum entries of $\mathbf{D}_1$ and $\mathbf{D}_2$, respectively.
 Therefore, $R$ is maximized. As a result, $\mathbf{S_u}$ and $\mathbf{S_v}$ have the same eigenvectors.
\end{proof}

We then introduce the interlacing theorem in Lemma \ref{LE_interlacing_theorem} \cite[p.182]{Horn_matrix_analysis} which will be used in proving beamforming is optimal (Theorem \ref{Th_optimal_BF}).
\begin{lemma}[Interlacing theorem] \label{LE_interlacing_theorem}
Let $\mathbf{M}\in \mathds{C}^{n\times n}$ be a Hermitian matrix and let $\mathbf{a}\in \mathds{C}^{n}$ be a given vector. We then have
\begin{align}
&\mbox{(a)}\,\,\,\,\,\,\lambda_{k}(\mathbf{M}\pm\mathbf{a}\mathbf{a}^H) \leq \lambda_{k+1}(\mathbf{M}) \leq \lambda_{k+2}(\mathbf{M}\pm\mathbf{a}\mathbf{a}^H),\,\,\,k=1,2,\ldots,n-2,\label{EQ_interlacing1}\\
&\mbox{(b)}\,\,\,\,\,\,\lambda_{k}(\mathbf{M}) \leq \lambda_{k+1}(\mathbf{M}\pm\mathbf{a}\mathbf{a}^H) \leq \lambda_{k+2}(\mathbf{M}),\,\,\,k=1,2,\ldots,n-2,\label{EQ_interlacing2}
\end{align}
where $\lambda_k(\mathbf{A})$ is the $k$th eigenvalue of $\mathbf{A}$ in ascending order.
\end{lemma}

First, we identify the rank property of the optimal
$\mathbf{S_u}^*$.

\begin{Theo}\label{Th_optimal_BF}
For the MISOSE fast fading wiretap channel with the perfect
information of the legitimate channel $\mathbf{h}$, and only the
statistics of the eavesdropper channel $\mathbf{g}\sim CN(0,
\mathbf{I}_{n_T})$ known at the transmitter, with the proposed
GAN-BF, the optimal covariance matrix of signal for \eqref{EQ_general_rate_AN} is $\mathbf{S_u^*}=\frac{P_U}{||\mathbf{h}||^2}\mathbf{h}\mathbf{h}^H$.
\end{Theo}
\begin{proof} Since the secrecy rate
optimization problem \eqref{EQ_general_rate_AN} is non-convex, we
can use the Karush-Kuhn-Tucker (KKT) conditions to find the
necessary conditions for the optimal solutions.
 We first transform
\eqref{EQ_general_rate_AN} into the following form to simplify the
KKT conditions
\begin{align}\label{EQ_general_rate_AN2}
R_{GAN}=\left(\max_{\mathbf{S_u},\,\mathbf{S_v}\in \textsl{S}}\log
\left(\frac{1+\mathbf{h}^H\left(\mathbf{S_u}+\mathbf{S_v}\right)\mathbf{h}}{1+\mathbf{h}^H\mathbf{S_v}\mathbf{h}}\right)
-\E\left[\log
\left(\frac{1+\mathbf{g}^H\left(\mathbf{S_u}+\mathbf{S_v}\right)\mathbf{g}}{1+\mathbf{g}^H\mathbf{S_v}\mathbf{g}}\right)\right]\right)^+.
\end{align}
Compared with \eqref{EQ_general_rate_AN}, in
\eqref{EQ_general_rate_AN2}, we place the maximum inside the
operation $(.)^+$. The equivalence of \eqref{EQ_general_rate_AN}
and \eqref{EQ_general_rate_AN2} comes from the fact that we can
represent $R_{GAN}$ by range of the objective inside $()^+$ in
\eqref{EQ_general_rate_AN} as the union of the sets of positive
and negative rates $R^+$ and $R^-$, respectively, as
$R_{GAN}=\max(R^+\bigcup R^-)^+=\max(R^+,\,R^-)^+$, which is
$\max(R^+)$ when $R^+$ is a nonempty set and zero, otherwise. On
the other hand, $(\max(R^+\bigcup R^-))^+$ is also $\max(R^+)$
when
 $R^+$ is a nonempty set and zero, otherwise. Thus we know \eqref{EQ_general_rate_AN} and
 \eqref{EQ_general_rate_AN2} are equivalent.
Let $\lambda\geq 0$, $\mathbf{\psi_u}\succeq 0$, and
$\mathbf{\psi_v}\succeq 0$ be the Lagrange multipliers of the three
constraints in \eqref{EQ_power_constraint}, respectively, the KKT
conditions of \eqref{EQ_general_rate_AN} is
\begin{align}
\bm\Theta_1=&\mathbf{S_u^*}=\mathbf{A}(\mathbf{S_u^*},\mathbf{S_v^*})-\lambda
\mathbf{I}_{n_T}+\mathbf{\psi_u}^T=\mathbf{0}, \label{EQ_KKT1}\\
\bm\Theta_2=&\mathbf{S_v^*}=\mathbf{A}(\mathbf{S_u^*},\mathbf{S_v^*})-\frac{\mathbf{h}\mathbf{h}^H}{1+\mathbf{h}^H\mathbf{S_v^*}\mathbf{h}}+\E\left[\frac{\mathbf{g}\mathbf{g}^H}{1+\mathbf{g}^H\mathbf{S_v^*}\mathbf{g}}\right]
-\lambda\mathbf{I}_{n_T}+\mathbf{\psi_v}^T=\mathbf{0}, \label{EQ_KKT2}\\
\mathbf{\psi_u}\mathbf{S_u^*}=&\mathbf{S_u^*}\mathbf{\psi_u}=\mathbf{0}, \label{EQ_KKT3}\\
\mathbf{\psi_v}\mathbf{S_v^*}=&\mathbf{S_v^*}\mathbf{\psi_v}=\mathbf{0}, \\
\tr{\left(\mathbf{S_u^*}+\mathbf{S_v^*}\right)}\leq&P_T, \,\,\mathbf{S_u^*}\succeq\mathbf{0}, \,\,\mathbf{S_v^*}\succeq\mathbf{0},
\end{align}
where
\begin{align}\label{EQ_A}
\mathbf{A}(\mathbf{S_u^*},\mathbf{S_v^*})&\triangleq \mathbf{a}\mathbf{a}^H+\mathbf{M},\\
\mathbf{a}\mathbf{a}^H&\triangleq\frac{\mathbf{h}\mathbf{h}^H}{1+\mathbf{h}^H\left(\mathbf{S_u^*}+\mathbf{S_v^*}\right)\mathbf{h}},\\
\mathbf{M}&\triangleq-\E\left[\frac{\mathbf{g}\mathbf{g}^H}{1+\mathbf{g}^H\left(\mathbf{S_u^*}+\mathbf{S_v^*}\right)\mathbf{g}}\right],
\end{align}
and $\mathbf{S_u^*}$ and $\mathbf{S_v^*}$ are the optimal input
covariance matrices of $\mathbf{u}$ and $\mathbf{v}$, respectively. In the following we denote $\mathbf{A}(\mathbf{S_u^*},\mathbf{S_v})$ by $\mathbf{A}^*$ to simplify the notation. After left and right multiplying \eqref{EQ_KKT1} by $(\mathbf{S_u^*})^T$, with \eqref{EQ_KKT3}, we have the relation
$\mathbf{A^*}(\mathbf{S_u}^*)^T=(\mathbf{S_u}^*)^T\mathbf{A^*}=\lambda(\mathbf{S_u^*})^T$,
where
$\lambda=\frac{\tr{\left(\mathbf{A^*}(\mathbf{S_u}^*)^T\right)}}{\tr{\left((\mathbf{S_u}^*)^T\right)}}$.
Then we can apply \cite[Lemma 8]{Pulu_ergodic} to ensure $\lambda
>0$, if $R>0$. Since $\mathbf{A}^*$ and $(\mathbf{S_u^*})^T$ commute, they have
the same eigenvectors. Therefore, we have
\begin{equation}\label{EQ_KKT_equality}
\bm\Lambda_{\mathbf{A}^*}\bm\Lambda_{\mathbf{S_u^*}}=\bm\Lambda_{\mathbf{S_u^*}}\bm\Lambda_{\mathbf{A^*}}=\lambda\bm\Lambda_{\mathbf{S_u^*}},
\end{equation}
where $\Lambda_{\mathbf{A}^*}$ and $\Lambda_{\mathbf{S_u^*}}$ are the eigenvalue matrices of $\mathbf{A}^*$ and $\mathbf{S_u^*}$, respectively.
Due to $\mathbf{M}$ in \eqref{EQ_A} is a negative-definite matrix \cite[Lemma4]{Pulu_ergodic}, from Lemma \ref{LE_interlacing_theorem}, we know that all eigenvalues of
$\mathbf{A^*}$ are smaller to zero except for the largest one. This can be explained as following. By using Lemma \ref{LE_interlacing_theorem} and letting $k=n_T-2$ in \eqref{EQ_interlacing2}, we have $\lambda_{n_T-1}(\mathbf{A^*}) \leq \lambda_{n_T}(\mathbf{M})$. Note that $\mathbf{M}$ is a negative definite matrix, i.e., $\lambda_{n_T} (\mathbf{M})<0$. So we have $\lambda_{i}(\mathbf{A^*})<0$ for $i=1,2,\ldots,n_T-1$.
Since $\lambda$ is positive, from \eqref{EQ_KKT_equality} we know that it must be the largest eigenvalue of $\mathbf{A^*}$, i.e. $\lambda =\lambda_{n_T}(\mathbf{A^*})$. In order to make the equality
$\bm\Lambda_{\mathbf{A^*}}\bm\Lambda_{\mathbf{S_u^*}}=\lambda\bm\Lambda_{\mathbf{S_u^*}}$ valid,
the eigenvalues of $\mathbf{S_u^*}$ corresponding to non-positive
eigenvalues of $\mathbf{ A^*}$ must be all zeros. Therefore, we obtain that $\mathbf{S_u^*}$ has only one nonzero eigenvalue. So the
covariance matrix of $\mathbf{S_u^*}$ is rank one if $R>0$. Then with Lemma \ref{Th_U_eq_V}, we conclude the proof.
\end{proof}

In the following we prove an important property, that is, using all the power is optimal for the proposed AN scheme.
\begin{lemma}\label{Lemma_equality_power_constraint}
To maximize \eqref{EQ_general_rate_AN}, the sum power constraint in \eqref{EQ_power_constraint} is hold with equality.
\end{lemma}
\begin{proof}
Similar to Theorem 2, the key observation here is
that with the selection of eigenvectors of signal and AN in Lemma \ref{Th_U_eq_V}, the first term on the RHS of (\ref{EQ_origin_Th2_problem}) is
independent of the power of AN in the null space of the legitimate channel. Thus to find $P_{V_i}$ for $i=2,3,\ldots,n_T$ given $P_U$ and $P_{V_1}$, the
objective function becomes
\begin{align}\label{EQ_min_2nd_term}
\underset{P_{V_2}\cdots P_{V_{n_T}}}{\min}\, \E\left[ \log\left(1+\frac{\tilde{G}_1P_{U}}{1+\tilde{G}_1P_{V_1}+\underset{i=2}{\overset{n_T}{\sum}}\tilde{G}_iP_{V_i}}\right) \right].
\end{align}
From \eqref{EQ_min_2nd_term} it can be easily seen that given
 $P_U$ and $P_{V_1}$, the value of the objective function decreases with increasing $P_T$. Thus we may change the first inequality constraint in
 \eqref{EQ_power_constraint} as an equality one.
\end{proof}

Based on Lemma \ref{Th_U_eq_V} and \ref{Lemma_equality_power_constraint}, we have the following property for AN.

\begin{lemma}\label{Th_AN_null_space_uniform}
For the optimization problem \eqref{EQ_general_rate_AN}, the optimal covariance matrix of AN is
\[
\mathbf{S_v}^*=\frac{1}{n_T-1}\left(\frac{n_TP_{V_1}-P_T+P_U}{||\mathbf{h}||^2}\mathbf{h}\mathbf{h}^H+(P_T-P_U-P_{V_1})\mathbf{I}\right).
\]

\end{lemma}

\begin{proof}To proceed, we transform \eqref{EQ_min_2nd_term} as
\begin{align}
\underset{P_{V_2}\cdots P_{V_{n_T}}}{\max}\, \E\left[ \log\left(1+\tilde{G}_1P_{V_1}+\underset{i=2}{\overset{n_T}{\sum}}\tilde{G}_iP_{V_i}\right)-\log\left(1+\tilde{G}_1(P_U+P_{V_1})+\underset{i=2}{\overset{n_T}{\sum}}\tilde{G}_iP_{V_i}\right) \right]=\underset{P_{V_2}\cdots P_{V_{n_T}}}{\max}\, \E_{\tilde{G}_1}\left[f(x)\Big|\tilde{G}_1\right],\label{EQ_obj_AN_null}
\end{align}
where the equality comes from the conditional mean, $f(x)\triangleq\E\left[ \log\left(a+x \right)-\log\left(b+x\right) \right]$
and we denote $1+\tilde{G}_1P_{V_1}$, $1+\tilde{G}_1(P_U+P_{V_1})$, and $\underset{i=2}{\overset{n_T}{\sum}}\tilde{G}_iP_{V_i}$
by $a$, $b$, and $x$, respectively. If given $\tilde{G}_1=g_1,\,\forall g_1$, the optimal power allocation of $f(x)$ is $P_{V_2}=P_{V_3}=\cdots =P_{V_{n_T}}$,
then for the problem on the left hand side (LHS) of \eqref{EQ_obj_AN_null},  this power allocation is also optimal. This is due to the fact that $\tilde{G}_i$ is
 unknown at transmitter by whom can not be used to change the power allocation. Therefore, we want to prove that under $\underset{i=2}{\overset{n_T}{\sum}}P_{V_i}=P_T-P_U-P_{V_1}$
\begin{equation}\label{EQ_new_obj_AN_uniform}
f\left(\frac{P_T-P_U-P_{V_1}}{n_T-1}\underset{i=2}{\overset{n_T}{\sum}}\tilde{G}_i\right)\geq f\left(\underset{i=2}{\overset{n_T}{\sum}}\tilde{G}_iP_{V_i}\right),\,\,\forall P_{V_i},\,i=2,\,\cdots,\,n_T.
\end{equation}

Here we introduce some results from the
\textit{stochastic ordering theory} \cite{shaked_stochastic_order}
to prove the desired result.
\begin{definition}\cite[p.234]{shaked_stochastic_order}\label{Def_completely_mono} A function $\psi:[0,\infty)\rightarrow \mathds{R}$ is completely monotone if for all $x>0$ and $n=0,1,2,\cdots,$ its derivative $\psi^{(n)}$ exists and $(-1)^n\psi^{(n)}(x)\geq 0$.
\end{definition}

\begin{definition}\cite[(5.A.1)]{shaked_stochastic_order} \label{Def_LT}
Let $B_1$ and $B_2$ be two nonnegative random variables such that
$\E[e^{-sB_1}]\geq\E[e^{-sB_2}]$, for all $s>0$. Then $B_1$ is
said to be smaller than $B_2$ in the Laplace transform order,
denoted as $B_1\leq_{LT} B_2$.
\end{definition}

\begin{lemma}\cite[Th. 5.A.4]{shaked_stochastic_order} \label{Lemma_LT_eq_MG}
Let $B_1$ and $B_2$ be two nonnegative random variables. If
$B_1\leq_{LT}B_2$ then $\E[f(B_1)]\leq\E[f(B_2)]$, where the first
derivative $\psi$ of a differentiable function $f$ on $[0,\infty)$
is completely monotone, provided that the expectations exist.
\end{lemma}

To prove \eqref{EQ_new_obj_AN_uniform}, we let
$B_1=\underset{i=2}{\overset{n_T}{\sum}}\tilde{G}_iP_{V_i}$,
$B_2=\underset{i=2}{\overset{n_T}{\sum}}\tilde{G}_iP_{V_i}^*$ to invoke Lemma \ref{Lemma_LT_eq_MG}, where $P_{V_i}^*$ denotes the optimal value of $P_{V_i}$.
It can be easily verified that $\psi$(x), the first derivative of
$f(x)$, satisfies Definition \ref{Def_completely_mono}. More
specifically, the $n$th derivative of $\psi$ meets
\begin{equation}
\psi^{(n)}(x)=\left\{\begin{array}{ll}
\frac{n!}{(a+x)^{n+1}}-\frac{n!}{(b+x)^{n+1}}>0, & \mbox{if $n$ is even,}\\
\frac{-n!}{(a+x)^{n+1}}+\frac{n!}{(b+x)^{n+1}}<0, & \mbox{if $n$ is odd,}\\
\end{array}\right.
\end{equation}
when $x>0$, since by definition, $b>a >0$ when $R>0$. Now from Lemma
\ref{Lemma_LT_eq_MG} and Definition \ref{Def_LT}, we know that to
prove \eqref{EQ_new_obj_AN_uniform} is equivalent to proving
$\E[e^{-sB_1}]\geq \E[e^{-sB_2}]$ or $\log
(\E[e^{-sB_1}]/\E[e^{-sB_2}])\geq 0,\,\forall s>0$. From
\cite[p.40]{Mathai}, we know that
\begin{equation}\label{EQ_log_MG}
\log
\left(\frac{\E[e^{-sB_1}]}{\E[e^{-sB_2}]}\right)=\sum_{k=2}^{n_T}\log(1+2P_{V_k}^*s)-\sum_{k=2}^{n_T}\log(1+2P_{V_k}s).
\end{equation}
To show the above is nonnegative, we resort to the majorization
theory \cite{Marshall_Inequalities}. Note that
$\sum_{k=2}^{n_T}\log(1+2\check{P}_{V_k}s)$ is a Schur-concave
function in $(\check{P}_{V_2},\ldots,\check{P}_{V_{n_T}})$, $\forall s>0$,
and by the definition of majorization
\[
(P_{V_2}^*,\,\cdots,P_{V_{n_T}}^*)=\left(\frac{P_T-P_U-P_{V_1}}{n_T-1},\,\frac{P_T-P_U-P_{V_1}}{n_T-1},\,\cdots,\,\frac{P_T-P_U-P_{V_1}}{n_T-1}\right)\prec(P_{V_2},\,\cdots,P_{V_{n_T}}),
\]
we know that the RHS of \eqref{EQ_log_MG} is nonnegative, $\forall
s>0$. Then \eqref{EQ_new_obj_AN_uniform} is valid.
From Lemma \ref{Th_U_eq_V} and \ref{Th_AN_null_space_uniform}, we can conclude that
\begin{align}
\mathbf{S_v^*}&=\left[\mathbf{h/||h||,\,h^{\perp}/||h||}\right]diag\left(P_{V_1},\,\frac{P_T-P_U-P_{V_1}}{n_T-1},\,\cdots,\frac{P_T-P_U-P_{V_1}}{n_T-1}\right)\left[\mathbf{h/||h||,\,h^{\perp}/||h||}\right]^H.
\end{align}
Then with the expansion
\[
\frac{\mathbf{h}\mathbf{h}^H}{||\mathbf{h}||^2}+\frac{\mathbf{h}^{\perp}(\mathbf{h}^{\perp})^H}{||\mathbf{h}||^2}=\mathbf{I},
\]
we conclude the proof.
\end{proof}
After substituting the $\mathbf{S_u^*}$ from Theorem
\ref{Th_optimal_BF} and $\mathbf{S_v^*}$ from Lemma
\ref{Th_AN_null_space_uniform} into \eqref{EQ_general_rate_AN}, we
can get \eqref{EQ_simplified_obj}. Note that when the main channel
is fast faded but perfectly known at transmitter, as
\cite{Li_fading_secrecy_j}, the achievable secrecy rate for this
setting can be easily obtained from results in Theorem
\ref{Th_main_result}.

\section{The iterative algorithm for power allocations between signal and \\ generalized artificial noise}\label{Sec_proposed algorithm}
Although we have simplified the optimization problem in
\eqref{EQ_general_rate_AN} to \eqref{EQ_simplified_obj},  since
\eqref{EQ_simplified_obj} is a non-convex stochastic optimization
problem, it is still difficult to analytically solve the optimal
power allocation $P_U$, $P_{V_1}$, and $P_{V_2}$ in
\eqref{EQ_simplified_obj}. Thus in this section we propose an
iterative power allocation algorithm summarized in Table I, which
can find solutions almost the same as the brute-force
search. However, the complexity of the proposed algorithm
is much lower than the one based on brute-force search. More specifically, the brute force search requires searching on a plane for the three variables $P_U$, $P_{V_1}$, and $P_{V_2}$, simultaneously. However, the proposed algorithm divide the search into two sub-problems which costs much less complexity. Before introducing the iterative algorithm, we first
provide a necessary condition in Theorem
\eqref{Th_necessary_condition} for the optimal covariance matrix
$\bS^*_\bv$ of the GAN to be full rank. This condition will be
useful to test the correctness of power allocation found in
proposed algorithm.

First define \begin{align}
F_k\left(x\right)=\int^{\infty}_0\frac{xe^{-t}}{\left(1+xt\right)^k}dt
=e^{1/x}E_k\left(1/x\right),\notag
\end{align}
where $E_k(x)$ is the En-function \cite{Stegun_handbook}.\\

Then we have the necessary condition in the following.
\begin{Theo}\label{Th_necessary_condition}
The necessary condition for the power allocation
$(P_U,P_{V_1}, P_{V_2})$ to be optimal for
\eqref{EQ_simplified_obj} is
\begin{align}
&\frac{1}{1+||\mathbf{ h}||^2P_{V_1}}-\frac{1+||\mathbf{ h}||^2P_U}{1+||\mathbf{ h}||^2(P_U+P_{V_1})}+\left(1+\frac{P_{V_2}}{P_{V_1}}\right)A_1 F_1(P_{V_1})+A_2 F_2(P_{V_1})\notag\\
&+\left(n_T-1+\frac{P_{V_1}}{P_{V_2}}\right)\underset{k=1}{\overset{n_T}{\sum}}\frac{B_k}{P_{V_2}}F_k(P_{V_2})-\frac{P_{V_1}}{P_{V_2}}B_{n_T}F_{n_T}(P_{V_2})
-(P_{V_1}+(n_T-1)P_{V_2})\frac{A_1^{'}}{P_U+P_{V_1}}F_1(P_U+P_{V_1})\notag\\
&-\frac{P_{V_1}A_2^{'}}{P_U+P_{V_1}}F_2(P_U+P_{V_1})-\left(n_T-1+\frac{P_{V_1}}{P_{V_2}}\right)\underset{k=1}{\overset{n_T}{\sum}}B_k^{'}F_k(P_{V_2})+\frac{P_{V_1}}{P_{V_2}}B_{n_T}^{'}F_{n_T}(P_{V_2})\gtrless 0,
\end{align}
then
\begin{align}
&\left(\frac{A_1}{P_{V_1}}F_1(P_{V_1})+\frac{A_2}{P_{V_1}}F_2(P_{V_1})\right)+\underset{k=1}{\overset{n_T-1}{\sum}}\frac{B_k}{P_{V_2}}F_k(P_{V_2})
-\frac{A_1^{'}}{P_U+P_{V_1}}F_1(P_U+P_{V_1})-\frac{A_2^{'}}{P_U+P_{V_1}}F_2(P_U+P_{V_1})\notag\\
&-\underset{k=1}{\overset{n_T-1}{\sum}}\frac{B_k^{'}}{P_{V_2}}F_k(P_{V_2})
\gtrless \frac{\vectornorm{\mathbf{h}}^4P_U}{\left(1+\vectornorm{\mathbf{h}}^2P_{V_1}\right)\left(1+\vectornorm{\mathbf{h}}^2\left(P_U+P_{V_1}\right)\right)},\label{EQ_Ineq_TH5}
\end{align}
where
\begin{align}
A_1 &= \frac{1-n_T}{\left( 1-\frac{P_{V_2}}{P_{V_1}} \right) ^{n_T}}\frac{P_{V_2}}{P_{V_1}},
\,A_2 = \frac{1}{\left( 1-\frac{P_{V_2}}{P_{V_1}} \right) ^{n_T-1}},\,
B_k = \frac{(n_T-k)\left( -\frac{P_{V_1}}{P_{V_2}}\right) ^{n_T-1-k}}{\left( 1-\frac{P_{V_2}}{P_{V_1}} \right) ^{n_T-k+1}},\notag\\
A_1^{'} &= \frac{1-n_T}{\left( 1-\frac{P_{V_2}}{P_U+P_{V_1}} \right) ^{n_T}}\frac{P_U+P_{V_1}}{P_{V_2}},\,
A_2^{'} = \frac{1}{\left( 1-\frac{P_{V_2}}{P_U+P_{V_1}} \right) ^{n_T-1}},\,
B_k^{'} = \frac{(n_T-k)\left( -\frac{P_U+P_{V_1}}{P_{V_2}}\right) ^{n_T-1-k}}{\left( 1-\frac{P_{V_2}}{P_U+P_{V_1}} \right) ^{n_T-k+1}},
\end{align}
with the requirement $ P_{V_1}>0$.
\end{Theo}

Now we present the derivation for the proposed iterative
algorithm. The key idea of the proposed algorithm is as following. To prevent the high complexity of simultaneously solving $P_U$, $P_{V_1}$, and $P_{V_2}$, we try to divide the problem as smaller ones and we can simply use bisection method to solve them. More specifically, we start from the KKT conditions, by eliminating the Lagrange multipliers, we form two equations each has different variables to solve. Then iteratively solve these two equations, we can find the power allocation. With the
Lagrange multipliers $\lambda \geq 0$, $\mu \geq 0$, $\mu_1 \geq
0$, and $\mu_2 \geq 0$, by the KKT conditions of
\eqref{EQ_simplified_obj}, we then have
\begin{align}
g_1\triangleq &\frac{\vectornorm{\mathbf{h}}^2}{1+\vectornorm{\mathbf{h}}^2(P_U^*+P_{V_1}^*)}- \E\left[ \frac{\tilde{G}_1}{1+(P_U^*+P_{V_1}^*)\tilde{G}_1+P_{V_2}\underset{i=2}{\overset{n_T}{\sum}}\tilde{G}_i} \right]-\lambda+\mu=0,\label{par_P_U}\\
g_2\triangleq&\frac{\vectornorm{\mathbf{h}}^2}{1+\vectornorm{\mathbf{h}}^2(P_U^*+P_{V_1}^*)}- \frac{\vectornorm{\mathbf{h}}^2}{1+\vectornorm{\mathbf{h}}^2P_{V_1}^*}\notag\\
&-\E\left[ \frac{\tilde{G}_1}{1+(P_U^*+P_{V_1}^*)\tilde{G}_1+P_{V_2}\underset{i=2}{\overset{n_T}{\sum}}\tilde{G}_i} \right]+\E\left[ \frac{\tilde{G}_1}{1+P_{V_1}^*\tilde{G}_1+P_{V_2}^*\underset{i=2}{\overset{n_T}{\sum}}\tilde{G}_i} \right]-\lambda+\mu_1=0,\label{par_P_V1}\\
g_3\triangleq&-\E\left[ \frac{\underset{i=2}{\overset{n_T}{\sum}}\tilde{G}_i} {1+(P_U^*+P_{V_1}^*)\tilde{G}_1+P_{V_2}^*\underset{i=2}{\overset{n_T}{\sum}}\tilde{G}_i} \right]+\E\left[ \frac{\underset{i=2}{\overset{n_T}{\sum}}\tilde{G}_i} {1+P_{V_1}^*\tilde{G}_1+P_{V_2}^*\underset{i=2}{\overset{n_T}{\sum}}\tilde{G}_i} \right]-(n_T-1)\lambda+\mu_2=0,\label{par_P_V2}\\
\mu P_U^*=&0,\label{EQ_constraint_P_U}\\
\mu_1 P_{V_1}^*=&0,\label{EQ_constraint_P_V1}\\
\mu_2 P_{V_2}^*=&0.\label{EQ_constraint_P_V2}
\end{align}
Assume that $P_U^*$, $P_{V_1}^*$, and $P_{V_2}^*$ are all non-zeros. Combining \eqref{par_P_U}, \eqref{par_P_V1}, \eqref{EQ_constraint_P_U}, and \eqref{EQ_constraint_P_V1} we have
\begin{align} \label{iter_1}
f_1(P_{V_1}^*,P_{V_2}^*)&\triangleq\frac{P_U^*P_{V_1}^*g_1-P_U^*P_{V_1}^*g_2}{P_U^*P_{V_1}^*}=\frac{\vectornorm{\mathbf{h}}^2}{1+\vectornorm{\mathbf{h}}^2P_{V_1}^*}- \E\left[ \frac{\tilde{G}_1}{1+P_{V_1}^*\tilde{G}_1+P_{V_2}^*\underset{i=2}{\overset{n_T}{\sum}}\tilde{G}_i} \right]=0.
\end{align}
Similarly, combining \eqref{par_P_U}, \eqref{par_P_V2}, \eqref{EQ_constraint_P_U}, and \eqref{EQ_constraint_P_V2}, and using the fact that
\begin{align}
\E\left[ \frac{\underset{i=2}{\overset{n_T}{\sum}}\tilde{G}_i} {1+P_{V_1}^*\tilde{G}_1+P_{V_2}^*\underset{i=2}{\overset{n_T}{\sum}}\tilde{G}_i} \right]=(n_T-1)\E\left[\frac{\tilde{G}_2} {1+P_{V_1}^*\tilde{G}_1+P_{V_2}^*\underset{i=2}{\overset{n_T}{\sum}}\tilde{G}_i} \right],
\end{align}
since the channel gain of each antenna is independent and identically distributed (i.i.d.),
we have
\begin{align} \label{iter_2}
f_2(P_U^*,P_{V_1}^*,P_{V_2}^*)\triangleq&\frac{P_U^*P_{V_2}^*g_1-P_U^*P_{V_2}^*\frac{1}{n_T-1}g_3}{P_U^*P_{V_2}^*}\notag\\
=&\frac{\vectornorm{\mathbf{h}}^2}{1+\vectornorm{\mathbf{h}}^2(P_U^*+P_{V_1}^*)}- \E\left[ \frac{\tilde{G}_1}{1+(P_U^*+P_{V_1}^*)\tilde{G}_1+P_{V_2}^*\underset{i=2}{\overset{n_T}{\sum}}\tilde{G}_i} \right]\notag\\
&+\E\left[ \frac{\tilde{G}_2}
{1+(P_U^*+P_{V_1}^*)\tilde{G}_1+P_{V_2}^*\underset{i=2}{\overset{n_T}{\sum}}\tilde{G}_i}
\right]-\E\left[ \frac{\tilde{G}_2}
{1+P_{V_1}^*\tilde{G}_1+P_{V_2}^*\underset{i=2}{\overset{n_T}{\sum}}\tilde{G}_i}
\right]=0.
\end{align}
Now for the $i$th iteration, with a given $P^{(i)}_{V_1}$,
we can find new $(P_U,P_{V_2})$ such that
$f_2(P_U,P^{(i)}_{V_1},P_{V_2})=0$ according to \eqref{iter_2}. We
can set $P_{U}=(P_T-P_{V_2}-P^{(i)}_{V_1})/(n_T-1)$ then $f_2(P_U,P_{V_1},P_{V_2})$ becomes a function with only one variable $P_{V_2}$. We let the resulted $P_{V_2}$ as
$P^{(i+1)}_{V_2}$. Then with a given $P^{(i+1)}_{V_2}$, we can
numerically solve a new $P_{V_1}$ such that
$f_1(P_{V_1},P^{(i+1)}_{V_2})$=0 according to \eqref{iter_1}. We let the resulted $P_{V_1}$ as
$P^{(i+1)}_{V_1}$ and the iterative algorithm follows. The
bisection method can be used to perform the numerical search.

Based on the concept described above, we explain each step
in Table \ref{TA iterative steps} in detail. First, numerically
finding the tuple $(P_{V_1},P_{V_2},P_{U})$ which exactly meet the equality
\eqref{iter_1} (or \eqref{iter_2}) is very hard. Therefore we
relax \eqref{iter_1} and \eqref{iter_2} by inequalities
\begin{equation} \label{eq_iter_ineq}
 |f_1(P_{V_1},P_{V_2})|<\epsilon_1 \;\; \mbox{and} \;\;
|f_2(P_U,P_{V_1},P_{V_2})|<\epsilon_1,
\end{equation}
respectively, where $\epsilon_1$ is a small constant. Once the
values from the bisection search validate the above inequalities,
they are treated as the solutions of these inequalities. Together
with the iteration step described in the end of the previous
paragraph, we obtain Step 2 and 3 in Table \ref{TA iterative
steps}. Second, relaxing equalities \eqref{iter_1} and
\eqref{iter_2} to inequalities \eqref{eq_iter_ineq} make solutions
obtained depend on $\epsilon_1$ and may not satisfy the KKT
conditions. Also the expectations in functions $f_1$ and $f_2$
(\eqref{iter_1} and \eqref{iter_2}) are calculated numerically via
generation of the channel realizations. Thus as in Step 4 of Table
\ref{TA iterative steps}, we use the analytical results in Theorem
\ref{Th_necessary_condition} to verify the correctness of the
solutions. Finally, the initial values for the first iteration in
Step 1 are as follows. Note that two initial values are needed for
specifying the search region of the bisection method. For
initializing Step 2, the two initial values for $P_U$ are $0$ and
$P_T-P_{V_1}^{(i)}$, such that the corresponding values of function
$f_2$ will have opposite signs. And there exists at least one
solution in the interval $[0,P_T-P_{V_1}^{(i)}]$. By the same reason,
for initializing Step 3, the two initial values for $P_{V_1}^{(i)}$ are
$0$ and $P_T-P_{V_2}(n_T-1)$.  In the $i$th iteration, the search regions are $[0,P_T-P_{V_1}^{(i)}]$ and $[0, P_T-(n_T-1)P_{V_2}^{(i)}]$} for $f_2$ and $f_1$, respectively.

However, the bisection method may not always work for searching
solutions for $|f_2|<\epsilon_1$ in Step 2 of Table \ref{TA
iterative steps}. Note that for the initial value
$P_U=P_T-P_{V_1}^{(i)}$, $f_2(P_T-P_{V_1}^{(i)},P_{V_1}^{(i)},0)<0$ given $P_{V_1}^{(i)}$.
On the other hand, given $P_{V_1}^{(i)}$, there exist two cases for
$f_2$ at initial value $P_U=0$: one is that
$f_2(0,P_{V_1}^{(i)},P_{V_2}^{(i)})<0$ as depicted in Figure \ref{Fig_f2_char}
(a), and the other is $f_2(0,P_{V_1}^{(i)},P_{V_2}^{(i)})>0$ as depicted in
Figure \ref{Fig_f2_char} (b). In the later case, the bisection
method works. However, if the former case happens, the function
values have the same sign, and the bisection method does not work.
To solve this problem, we can use the \emph{golden section method}
\cite{Heath_scientific_computing}, which is a technique for
finding the maximum in the interval $[0, P_T-P_{V_1}^{(i)}]$, i.e., to
numerically find $\tilde{P}_U$ first such that given $P_{V_1}^{(i)}$,
$f_2(\tilde{P}_U,P_{V_1}^{(i)},P_{V_2}^{(i)})$ is positive. After that we can
follow the step 2 in Table I to solve $P_U$ in the interval
$[\tilde{P}_U, P_T-P_{V_1}^{(i)}]$. If the maximum of
$f_2(P_U,P_{V_1}^{(i)},P_{V_2}^{(i)})$ in the interval $[0, P_T-P_{V_1}^{(i)}]$ is
still negative, we know that there does not exist any $P_U$ in
this interval such that $f_2(P_U,P_{V_1}^{(i)},P_{V_2}^{(i)})=0$ given
$P_{V_1}^{(i)}$. In this case, we set $P_U=0$ as the solution of
$f_2(P_U,P_{V_1}^{(i)},P_{V_2}^{(i)})=0$ given $P_{V_1}^{(i)}$. From simulation
results, according to the iterative algorithm in Table \ref{TA
iterative steps}, the power $P_U^{(i)}$, $P_{V_1}^{(i)}$, and $P_{V_2}^{(i)}$ will
converge to the optimal solution $P_U^{*}$, $P_{V_1}^{*}$, and
$P_{V_2}^{*}$, respectively, which satisfy the KKT necessary
conditions.


%

\textit{Remark 1}: Note that in Section \ref{Sec_proposed
algorithm} we assume that $P_U,\,P_{V_1},\mbox{and }\,P_{V_2}$ are
all non-zeros to eliminate the multipliers. For channel conditions
under which low rank AN covariance matrix is optimal, the proposed
algorithm may have $P_{V_1}$ converge to a value approximately
zero. When this value is smaller than a predefined threshold
$\epsilon_2$, we claim that $P_{V_1}=0$ is optimal.

\section{Simulation results}\label{Sec_simulation}
In this section, we illustrate the performance gain of the proposed transmission scheme over
 Goel and Negi's scheme. We use a 2 by 1 by 1 channel as an example. Assume that the noise
  variances of Bob and Eve are normalized to 1. From (\ref{EQ_simplified_obj}) we know that
  the rate $R_{GAN}$ only depends on the norm of the main channel. Therefore, we use $||\mathbf{h}||^2=0.05,\,0.1,\mbox{ and }\,0.2$ to
   indicate different channel conditions in the simulation. For the statistics of the eavesdropper's channel, we set
   $\E[\tilde{{G}}_1]=\E[\tilde{{G}}_2]=1$. In Fig. \ref{Fig_h_005}, \ref{Fig_h_01}, and \ref{Fig_h_02},
    which correspond to $||\mathbf{h}||^2=0.05,\,0.1,\mbox{ and }\,0.2$, respectively, we compare the rates of Goel and Negi's scheme to that of our proposed signaling with the generalized AN. The blue and black curves represent searching the optimal power allocations exhaustively and by the proposed iterative algorithm, respectively. In the iterative algorithm, we set the iteration number $MAXIT$ as 20, $MAXCheck$ as 5, and $\epsilon_1=\epsilon_2=10^{-5}$. From Fig. \ref{Fig_h_005}, \ref{Fig_h_01}, and \ref{Fig_h_02}, we can easily see that the proposed generalized AN scheme indeed provides apparent rate gains over Goel and Negi's scheme in the moderate SNR regions. In addition, we can observe that the rate gains decrease with increasing $||\mathbf{h}||^2$, which is consistent with the results in \cite{Li_fading_secrecy_j}. We can also find that the value of $P_T$ which provides the largest rate gain also decreases with increasing $||\mathbf{h}||^2$. This is because AN in the signal direction provides much more rate gains when Bob's received SNR is relatively small compared to Eve's. Furthermore, the power allocations of the proposed iterative algorithm indeed converges to those by exhaustive search. In and Fig. \ref{Fig_convergence01} we show the convergence rate of the proposed algorithm under $||\mathbf{h}||^2=0.1$ with different $P_T$. It can be found that the proposed algorithm converges fast under different  $P_T$, i.e., it costs at most 7 iterations to the final value, which verifies the complexity of solving the power allocation is much lower than the full search.

As another example, we also illustrate the optimal power allocation among $P_U$, $P_{V_1}$, and $P_{V_2}$ under $||\mathbf{h}||^2=0.05$ in Fig. \ref{Fig_power_allocation_005}. It can be easily seen that as the received SNR increases, the power allocated to $P_{V_1}$ decreases and the rate gain over Goel and Negi's scheme also decreases.

\section{Conclusion}\label{Sec_conclusion}
In this paper we generalized Goel and Negi's artificial noise (AN)
for fast fading secure transmission with full knowledge of the main channel and only the statistics of the
eavesdropper's channel state information at the transmitter. Instead
of transmitting AN in the null space of the legitimate channel, we considered
injecting AN in all directions, including the direction for
conveying the dedicated messages. Our main result provides a highly simplified power allocation problem to describe the ergodic secrecy rate.
To attain it, we proved that for a multiple-input
single-output single-antenna-eavesdropper
 system with the proposed AN injecting scheme, the optimal transmission scheme is a beamformer which is
 aligned to the direction of the legitimate channel. In addition, we provided the necessary condition for the
 optimal covariance matrix of AN to be full rank. After characterizing the optimal eigenvectors of the covariance matrices of
 signal and AN, we also developed an algorithm to efficiently solve the non-convex power allocation problem.
 Through simulations, we verified that the proposed scheme
 outperforms Goel and Negi's AN scheme under certain channel conditions, especially when the legitimate channel is poor.

\section{Appendix}
Before proving Theorem \ref{Th_necessary_condition}, we first introduce the following lemma which will be used.

\begin{lemma}\label{LE_Y_positive_definite}
Given $\mathbf{D}_1 \succ \mathbf{D}_2$,
\begin{align}
\mathbf{Y}\triangleq\E\left[\frac{\mathbf{g}\mathbf{g}^H}{1+\mathbf{g}^H\mathbf{D}_2^H\mathbf{g}}
\right]-\E\left[\frac{\mathbf{g}\mathbf{g}^H}{1+\mathbf{g}^H\mathbf{D}_1^H\mathbf{g}}\right]
\succ 0.\label{EQ_Y}
\end{align}
\end{lemma}
\begin{proof}
We first write the expectation in \eqref{EQ_Y} in the following integral,
\begin{align}
\mathbf{Y}_{1,1} &=\int^{\infty}_0e^{-t}\frac{1}{\left(1+P_{V_1}t\right)^2}\frac{1}{\left(1+P_{V_2}t\right)^{n_T-1}}dt- \int^{\infty}_0e^{-t}\frac{1}{\left(1+\left(P_U+P_{V_1}\right)t\right)^2}\frac{1}{\left(1+P_{V_2}t\right)^{n_T-1}}dt\notag\\
&=\int^{\infty}_0e^{-t}\left(\frac{1}{\left(1+P_{V_1}t\right)^2}-\frac{1}{\left(1+\left(P_U+P_{V_1}\right)t\right)^2}\right)\frac{1}{\left(1+P_{V_2}t\right)^{n_T-1}}dt>0, \label{EQ_Y_11}
\end{align}
and
\begin{align}
\mathbf{Y}_{i,i} &=\int^{\infty}_0e^{-t}\frac{1}{1+P_{V_1}t}\frac{1}{\left(1+P_{V_2}t\right)^{n_T}}dt-\int^{\infty}_0e^{-t}\frac{1}{1+\left(P_U+P_{V_1}\right)t}\frac{1}{\left(1+P_{V_2}t\right)^{n_T}}dt\notag\\
&=\int^{\infty}_0e^{-t}\left(\frac{1}{1+P_{V_1}t}-\frac{1}{1+\left(P_U+P_{V_1}\right)t}\right)\frac{1}{\left(1+P_{V_2}t\right)^{n_T}}dt>0, \label{EQ_Y_ii}
\end{align}
for $i=2,3,\ldots,n_T$, and from \cite[Lemma 4]{Pulu_ergodic}, we know that the non-diagonal entries of both the first and second terms of $\mathbf{Y}$ in \eqref{LE_Y_positive_definite} are zeros, then $\mathbf{Y}_{i,j}=0$ for $i \neq j$. Therefore, we know that $\mathbf{Y}$ is a diagonal matrix and each diagonal entry from \eqref{EQ_Y_11} and \eqref{EQ_Y_ii} is larger than zero, which completes the proof.
\end{proof}

We now provide the proof of Theorem \ref{Th_necessary_condition}
\begin{proof}
We first rearrange \eqref{EQ_KKT2} as
\begin{align}
\bm\Theta_2 = \mathbf{C}-\lambda\mathbf{I}_{n_T}+\bm{\psi}_\mathbf{v}^T=\mathbf{0}\notag,
\end{align}
where
\begin{align}
\mathbf{C}\triangleq &\mathbf{U}\mathbf{Y}\mathbf{U}^H -\mathbf{c}\mathbf{c}^H, \label{EQ_def_C} \\
 \mathbf{Y} \triangleq & \E\left[\frac{\mathbf{U}^H\mathbf{g}\mathbf{g}^H\mathbf{U}}{1+\mathbf{g}^H\mathbf{U}\mathbf{D}_2\mathbf{U}^H\mathbf{g}} \right]-\E\left[\frac{\mathbf{U}^H\mathbf{g}\mathbf{g}^H\mathbf{U}}{1+\mathbf{g}^H\mathbf{U}\mathbf{D}_1\mathbf{U}^H\mathbf{g}}\right]
=\E\left[\frac{\mathbf{g}\mathbf{g}^H}{1+\mathbf{g}^H\mathbf{D}_2^H\mathbf{g}} \right]-\E\left[\frac{\mathbf{g}\mathbf{g}^H}{1+\mathbf{g}^H\mathbf{D}_1^H\mathbf{g}}\right], \label{EQ_Y_DD}\\
\mathbf{c}\triangleq&\left(\frac{\mathbf{h}^H\mathbf{S_u}\mathbf{h}}{\left(1+\mathbf{h}^H\mathbf{S_v^*}\mathbf{h}\right)\left(1+\mathbf{h}^H\left(\mathbf{S_u}+\mathbf{S_v^*}\right)\mathbf{h}\right)}\right)^{1/2}\mathbf{h}.\label{EQ_def_c}
\end{align}

Similar to \eqref{EQ_KKT_equality}, we have
\begin{equation}\label{EQ_KKT_equality2}
\bm\Lambda_{\mathbf{C}}\bm\Lambda_{\mathbf{S_v^*}}=\bm\Lambda_{\mathbf{S_v^*}}\bm\Lambda_{\mathbf{C}}=\mbox{tr}(\mathbf{C}\mathbf{S_v^*})\bm\Lambda_{\mathbf{S_v^*}}.
\end{equation}
And we know that the necessary condition for the optimal AN to be full rank is that when $\mbox{tr}(\mathbf{C}\mathbf{S_v^*})>0$, $\mathbf{C}$ does not have any negative eigenvalues; or, when $\mbox{tr}(\mathbf{C}\mathbf{S_v^*})<0$, $\mathbf{C}$ does not have any positive eigenvalues. To verify this property, we resort to the fact from \cite[Lemma 5]{Pulu_ergodic} that if all eigenvalues $\lambda$ of $\mathbf{a}\mathbf{a}^H-\mathbf{A}$ are negative, then $l(0)> 0$, where $l(\lambda)$ is defined as,
\begin{align}
l(\lambda)\triangleq1-\mathbf{a}^H\left(\mathbf{A}+\lambda
\mathbf{I}_{n_T}\right)^{-1}\mathbf{a},
\end{align}
and $\mathbf{A}\succ 0$. Note that $l(\lambda)$ is a strictly increasing function when $\lambda>0$. Note also that $\mathbf{C}$ in (\ref{EQ_def_C}) is negated of $\mathbf{a}\mathbf{a}^H-\mathbf{A}$. Thus all eigenvalues of $\mathbf{C}$ are positive implies $l(0)> 0$. Thus by substituting $\mathbf{c}$ and $\mathbf{U}\mathbf{Y}\mathbf{U}^H$ into $\mathbf{a}$ and $\mathbf{A}$, respectively, we have
\begin{align}\label{EQ_l_lambda}
l(\lambda)=1-\mathbf{c}^H\left(\mathbf{U}\mathbf{Y}\mathbf{U}^H+\lambda
\mathbf{I}_{n_T}\right)^{-1}\mathbf{c}.
\end{align}

 By Lemma \ref{LE_Y_positive_definite} we know $\left( \mathbf{U}\mathbf{Y}\mathbf{U}^H \right)^{-1}$ exists. Then we can expand $l(0)> 0$ from \eqref{EQ_l_lambda} as
\begin{align}\label{EQ_c_ineq}
 &1>\mathbf{c}^H\left( \mathbf{U}\mathbf{Y}\mathbf{U}^H \right)^{-1}\mathbf{c}.
 \end{align}
Then after substituting $\mathbf{c}$ from \eqref{EQ_def_c} to \eqref{EQ_c_ineq}, and using Theorem \ref{Th_optimal_BF} and Lemma \ref{Th_U_eq_V}, we have \begin{align}
\left[\mathbf{Y}^{-1}\right]_{1,1} < \frac{\left(1+\vectornorm{\mathbf{h}}^2P_{V_1}\right)\left(1+\vectornorm{\mathbf{h}}^2\left(P_U+P_{V_1}\right)\right)}{\vectornorm{\mathbf{h}}^4P_U}.\notag
\end{align}
From \cite[Lemma 4]{Pulu_ergodic} we know that $\mathbf{Y}$ is diagonal. In addition, with $\mathbf{Y}$ is invertible from the proof of Lemma \ref{LE_Y_positive_definite}, we can further rearrange the above as
\begin{align}
\left[\mathbf{Y}\right]_{1,1} > \frac{\vectornorm{\mathbf{h}}^4P_U}{\left(1+\vectornorm{\mathbf{h}}^2P_{V_1}\right)\left(1+\vectornorm{\mathbf{h}}^2\left(P_U+P_{V_1}\right)\right)}.\notag
\end{align}
Then by the definition of $\mathbf{Y}$ in \eqref{EQ_Y_DD}, and the fractional expansion, we can further express the above as
\begin{align}
&\left(\frac{A_1}{P_{V_1}}F_1(P_{V_1})+\frac{A_2}{P_{V_1}}F_2(P_{V_1})\right)\mathbf{1}_{P_{V_1}\neq 0}+\underset{k=1}{\overset{n_T-1}{\sum}}\frac{B_k}{P_{V_2}}F_k(P_{V_2})
-\frac{A_1^{'}}{P_U+P_{V_1}}F_1(P_U+P_{V_1})-\frac{A_2^{'}}{P_U+P_{V_1}}F_2(P_U+P_{V_1})\notag\\
&-\underset{k=1}{\overset{n_T-1}{\sum}}\frac{B_k^{'}}{P_{V_2}}F_k(P_{V_2})
> \frac{\vectornorm{\mathbf{h}}^4P_U}{\left(1+\vectornorm{\mathbf{h}}^2P_{V_1}\right)\left(1+\vectornorm{\mathbf{h}}^2\left(P_U+P_{V_1}\right)\right)},
\end{align}
where $A_1,\,A_2,\,A_1',\,A_2',\,B_k,\,\mbox{ and }B_k'$ for $k=1,2,\ldots,n_T-1$ are defined in the statement of the theorem.
In addition, $\mbox{tr}(\mathbf{C}\mathbf{S_v^*})>0$ implies
\begin{align}\label{EQ_intermediate_condition}
\frac{1}{1+||\mathbf{ h}||^2P_{V_1}}-\frac{1+||\mathbf{ h}||^2P_{U}}{1+||\mathbf{ h}||^2(P_U+P_{V_1})}+\mathbf{ E}\left[\frac{1+\mathbf{g}^H(\mathbf{ D}_1-\mathbf{ D}_2)\mathbf{ g}}{1+\mathbf{ g}^H\mathbf{ D}_1\mathbf{ g}}\right]-\mathbf{ E}\left[\frac{1}{1+\mathbf{ g}^H\mathbf{ D}_2\mathbf{ g}}\right]>0.
\end{align}
After some arrangement, \eqref{EQ_intermediate_condition} can be further represented by
\begin{align}
&\frac{1}{1+||\mathbf{ h}||^2P_{V_1}}-\frac{1+||\mathbf{ h}||^2P_U}{1+||\mathbf{ h}||^2(P_U+P_{V_1})}+\left(1+\frac{P_{V_2}}{P_{V_1}}\right)A_1 F_1(P_{V_1})+A_2 F_2(P_{V_1})+\left(n_T-1+\frac{P_{V_1}}{P_{V_2}}\right)\underset{k=1}{\overset{n_T}{\sum}}\frac{B_k}{P_{V_2}}F_k(P_{V_2})\notag\\
&-\frac{P_{V_1}}{P_{V_2}}B_{n_T}F_{n_T}(P_{V_2})
-(P_{V_1}+(n_T-1)P_{V_2})\frac{A_1^{'}}{P_U+P_{V_1}}F_1(P_U+P_{V_1})-\frac{P_{V_1}A_2^{'}}{P_U+P_{V_1}}F_2(P_U+P_{V_1})\notag\\
&-\left(n_T-1+\frac{P_{V_1}}{P_{V_2}}\right)\underset{k=1}{\overset{n_T}{\sum}}B_k^{'}F_k(P_{V_2})+\frac{P_{V_1}}{P_{V_2}}B_{n_T}^{'}F_{n_T}(P_{V_2})>0.
\end{align}
\end{proof}
\bibliographystyle{IEEEtran}

\renewcommand{\baselinestretch}{2}
\bibliography{IEEEabrv,SecrecyPs2}

\begin{thebibliography}{10}
\providecommand{\url}[1]{#1}
\csname url@rmstyle\endcsname
\providecommand{\newblock}{\relax}
\providecommand{\bibinfo}[2]{#2}
\providecommand\BIBentrySTDinterwordspacing{\spaceskip=0pt\relax}
\providecommand\BIBentryALTinterwordstretchfactor{4}
\providecommand\BIBentryALTinterwordspacing{\spaceskip=\fontdimen2\font plus
\BIBentryALTinterwordstretchfactor\fontdimen3\font minus
  \fontdimen4\font\relax}
\providecommand\BIBforeignlanguage[2]{{%
\expandafter\ifx\csname l@#1\endcsname\relax
\typeout{** WARNING: IEEEtran.bst: No hyphenation pattern has been}%
\typeout{** loaded for the language `#1'. Using the pattern for}%
\typeout{** the default language instead.}%
\else
\language=\csname l@#1\endcsname
\fi
#2}}

\bibitem{csiszar1978broadcast}
I.~Csisz{\'a}r and J.~Korner, ``{Broadcast channels with confidential
  messages},'' \emph{{IEEE} Trans. Inform. Theory}, vol.~24, no.~3, pp.
  339--348, 1978.

\bibitem{Wyner_wiretap}
A.~D. Wyner, ``The wiretap channel,'' \emph{Bell Syst. Tech. J.}, vol.~54, pp.
  1355--1387, 1975.

\bibitem{Secureconnect}
X.~Zhou, R.~K. Ganti, and J.~G. Andews, ``Secure wireless network connectivity
  with multi-antenna transmission,'' vol.~10, no.~2, pp. 425--430, Feb. 2011.

\bibitem{Shafiee_secrecy_2_2_1_J}
S.~Shafiee and S.~Ulukus, ``Towards the secrecy capacity of the {G}aussian
  {MIMO} wire-tap channel: the 2-2-1 channel,'' \emph{{IEEE} Trans. Inform.
  Theory}, vol.~55, no.~9, pp. 4033--4039, Sept. 2009.

\bibitem{Khisti_MIMOME}
A.~Khisti and G.~W. Wornell, ``Secure transmission with multiple antennas-{II}:
  The {MIMOME} wiretap channel,'' \emph{{IEEE} Trans. Inform. Theory}, vol.~56,
  no.~11, pp. 5515--5532, Nov 2010.

\bibitem{Oggier_MIMOME}
F.~Oggier and B.~Hassibi, ``The secrecy capacity of the {MIMO} wiretap
  channel,'' \emph{{IEEE} Trans. Inform. Theory}, vol.~57, no.~8, Aug. 2011.

\bibitem{Liu_MIMO_wiretap}
T.~Liu and S.~S. (Shitz), ``A note on the secrecy capacity of the
  multiple-antenna wiretap channel,'' vol.~55, no.~6, pp. 2547--2553, Jun.
  2009.

\bibitem{Liang_fading_secrecy}
Y.~Liang, V.~Poor, and S.~S. (Shitz), ``Secure communication over fading
  channels,'' \emph{{IEEE} Trans. Inform. Theory}, vol.~54, no.~6, pp.
  2470--2492, Jun. 2008.

\bibitem{Goel_AN}
S.~Goel and R.~Negi, ``Guaranteeing secrecy using artificial noise,''
  \emph{{IEEE} Trans. Wireless Commun.}, vol.~7, no.~6, pp. 2180--2189, June
  2008.

\bibitem{gopala2008secrecy}
P.~Gopala, L.~Lai, and H.~El~Gamal, ``{On the secrecy capacity of fading
  channels},'' \emph{{IEEE} Trans. Inform. Theory}, vol.~54, no.~10, pp.
  4687--4698, Oct. 2008.

\bibitem{Khisti_MISOME}
A.~Khisti and G.~W. Wornell, ``Secure transmission with multiple antennas-{I}:
  The {MISOME} wiretap channel,'' \emph{{IEEE} Trans. Inform. Theory}, vol.~56,
  no.~7, pp. 3088--3104, July 2010.

\bibitem{Li_fading_secrecy_j}
Z.~Li, R.~Yates, and W.~Trappe, ``Achieving secret communication for fast
  {R}ayleigh fading channels,'' \emph{{IEEE} Trans. Wireless Commun.}, vol.~9,
  no.~9, pp. 2792 -- 2799, Sep. 2010.

\bibitem{Pulu_ergodic}
J.~Li and A.~Petropulu, ``On ergodic secrecy rate for {G}aussian {MISO} wiretap
  channels,'' \emph{{IEEE} Trans. Wireless Commun.}, vol.~10, no.~4, pp.
  1176--1187, Apr. 2011.

\bibitem{Lin_Ergodic_secrecy_capacity}
S.~C. Lin and P.~H. Lin, ``On ergodic secrecy capacity of multiple input
  wiretap channel with statistical {CSIT},''
  \emph{http://arxiv.org/abs/1201.2868}, Jan. 2012.

\bibitem{Horn_matrix_analysis}
R.~A. Horn and C.~R. Johnson, \emph{Matrix analysis}.\hskip 1em plus 0.5em
  minus 0.4em\relax Cambridger University Press, 1985.

\bibitem{shaked_stochastic_order}
M.~Shaked and J.~G. Shanthikumar, \emph{Stochastic Orders}.\hskip 1em plus
  0.5em minus 0.4em\relax Springer, 2007.

\bibitem{Mathai}
A.~M. Mathai and S.~B. Provost, \emph{Quadratic forms in random
  variables}.\hskip 1em plus 0.5em minus 0.4em\relax Marcel Dekker, New York,
  1992.

\bibitem{Marshall_Inequalities}
A.~W. Marshall and I.~Olkin, \emph{Inequalities: theory of majorization and its
  application}.

\bibitem{Stegun_handbook}
M.~M.~Abramowitz and I.~A. I.~A.~Stegun, \emph{Handbook of Mathematical
  Functions with Formulas, Graphs, and Mathematical Tables}.\hskip 1em plus
  0.5em minus 0.4em\relax New York: Dover, 1972.

\bibitem{Heath_scientific_computing}
M.~T. Heath, \emph{Scientific computing: an introductory survey}, 2nd~ed.\hskip
  1em plus 0.5em minus 0.4em\relax McGraw-Hill.

\end{thebibliography}

\newpage

\begin{table} [ht]
\begin{center}
\caption{The iterative algorithm for power allocation between
signal and generalized AN}
\begin{tabular}{l l} \label{TA iterative steps}
Step 1 & Set $i=0$, $P_{V_1}^{(0)}=0$, and initialize search region for the bisection method.\\
Step 2 & Given $P_{V_1}^{(i)}$ and the total power constraint
\eqref{EQ_power_constraint}, find $P_{V_2}$ (and thus
$P_U=(P_T-P_{V_2}-P^{(i)}_{V_1})/(n_T-1)$) \\ & such that
$|f_2(P_U,P_{V_1}^{(i)},P_{V_2})|<\epsilon_1$, where $f_2$ is
defined in \eqref{iter_2}. \\ & Set $P_{V_2}^{(i+1)}=P_{V_2}$ \\
Step 3 & Given $P_{V_2}^{(i+1)}$ and the total power constraint
\eqref{EQ_power_constraint}, find $P_{V_1}$ \\ & such that $|
f_1(P_{V_1},P_{V_2}^{(i+1)})|<\epsilon_1$, where $f_1$ is defined
in \eqref{iter_1}\\ & Set $P_{V_1}^{(i+1)}=P_{V_1}$. \\
Step 4 & Let $i=i+1$ and repeat Step 2 to Step 3 until $MAXIT$.\\
Step 5 & Check the whether the final power allocations meet
Theorem \ref{Th_necessary_condition}. \\ & If not, randomly
re-initialize $P_{V_1}^{(0)}$ and run Step 1-4 until $MAXCheck$.
\end{tabular}
\end{center}
\end{table}
\newpage

\begin{figure}[htp]
\centering \epsfig{file=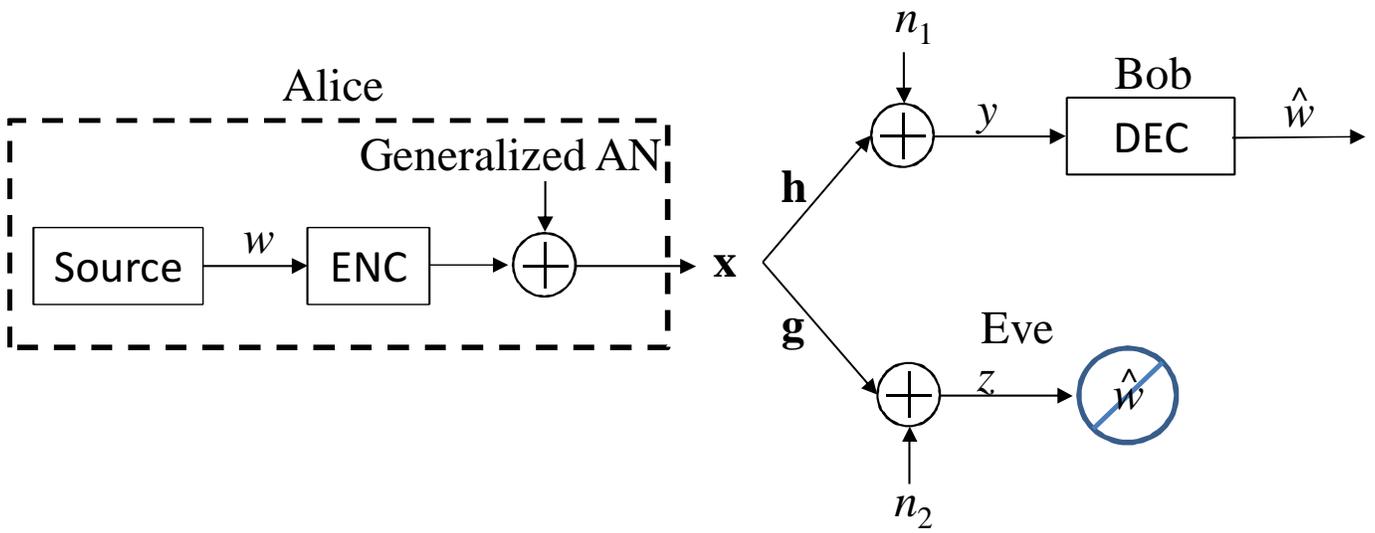 ,width=1\textwidth}
\caption{System model.} \label{Fig_system}
\end{figure}

\begin{figure}\vspace{4cm}
\centering \epsfig{file=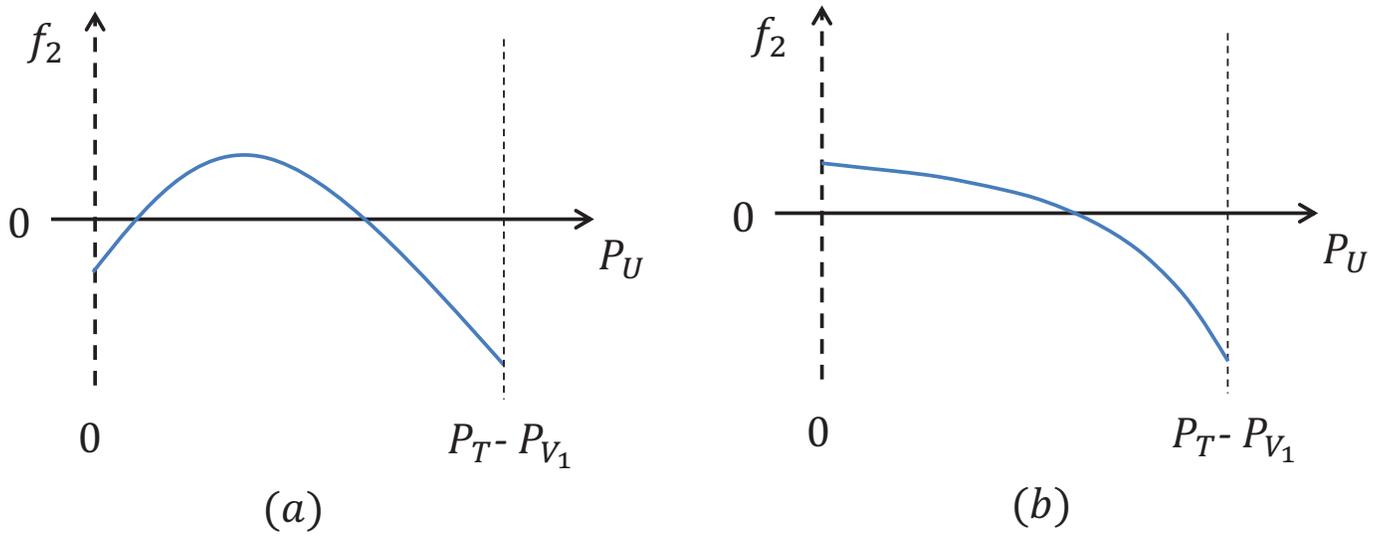 ,width=1\textwidth}
\caption{Characteristic of $f_2(P_U,P_{V_1},P_{V_2})$ given $P_{V_1}$.} \label{Fig_f2_char}
\end{figure}


\begin{figure}
\centering \epsfig{file=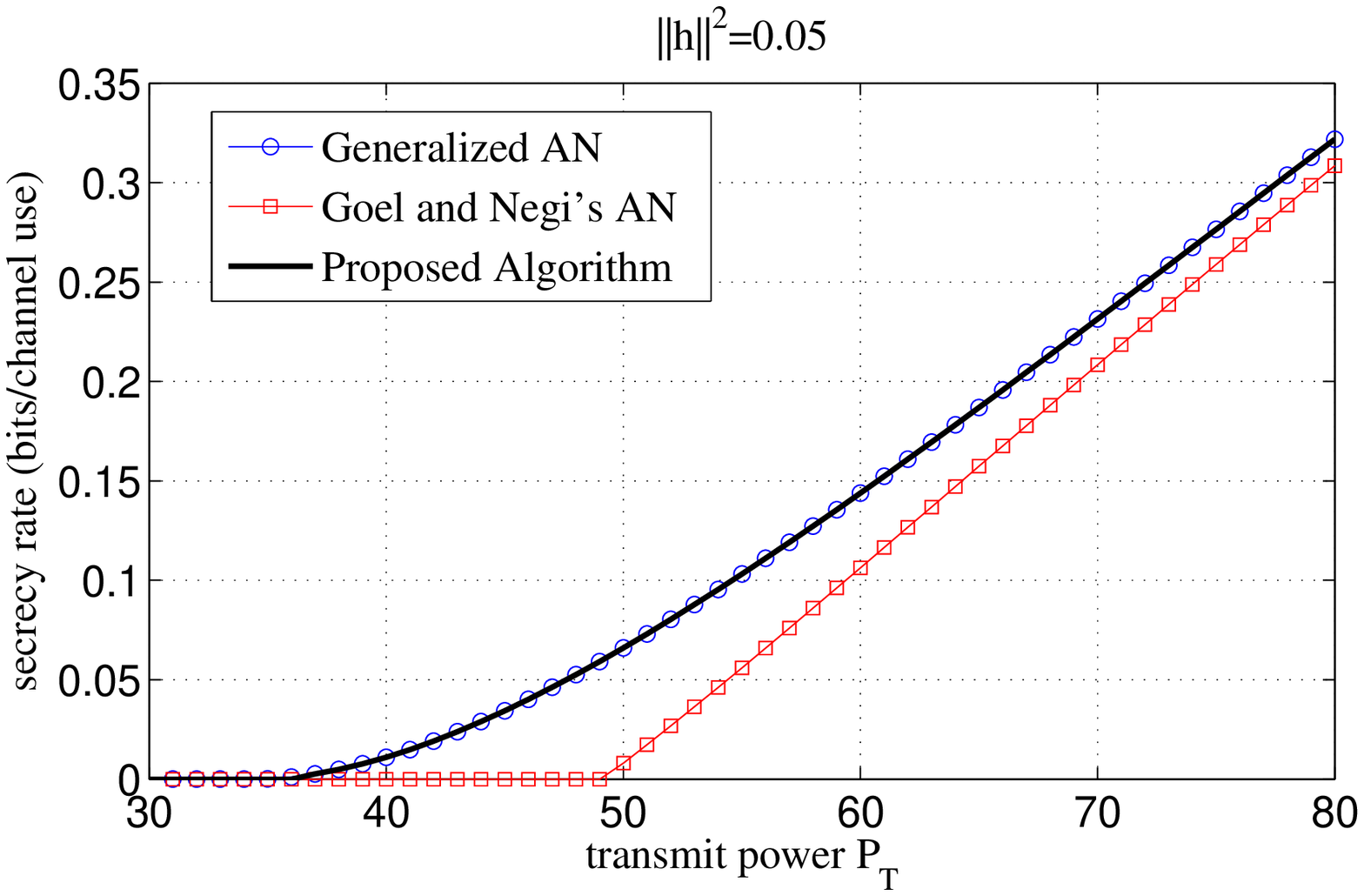 , width=1\textwidth}
2\caption{Secrecy rate versus transmit power under
$||\mathbf{h}||^2=0.05$.} \label{Fig_h_005}
\end{figure}

\begin{figure}
\centering \epsfig{file=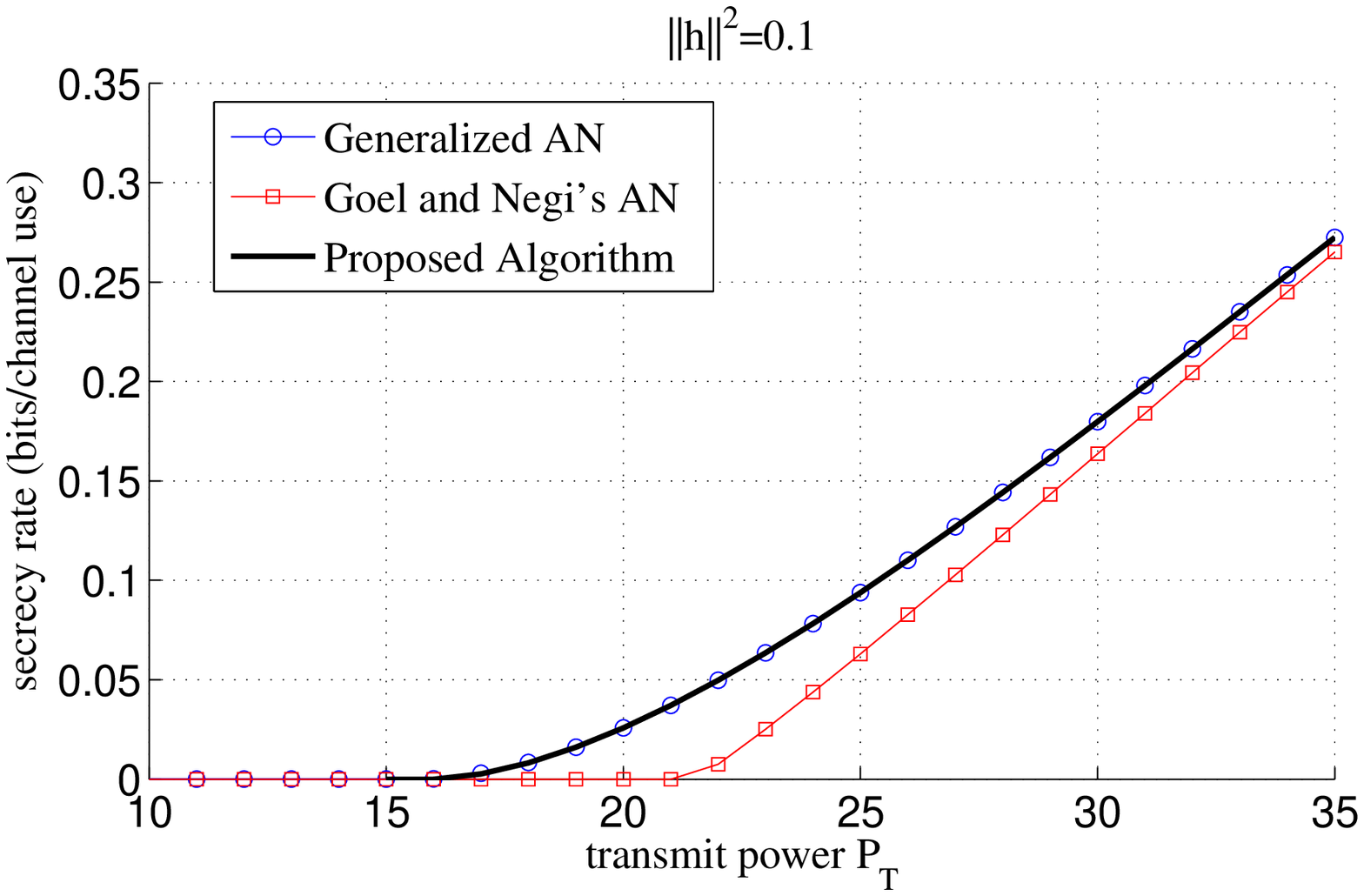 , width=1\textwidth}
\caption{Secrecy rate versus transmit power under
$||\mathbf{h}||^2=0.1$.} \label{Fig_h_01}
\end{figure}

\begin{figure}
\centering \epsfig{file=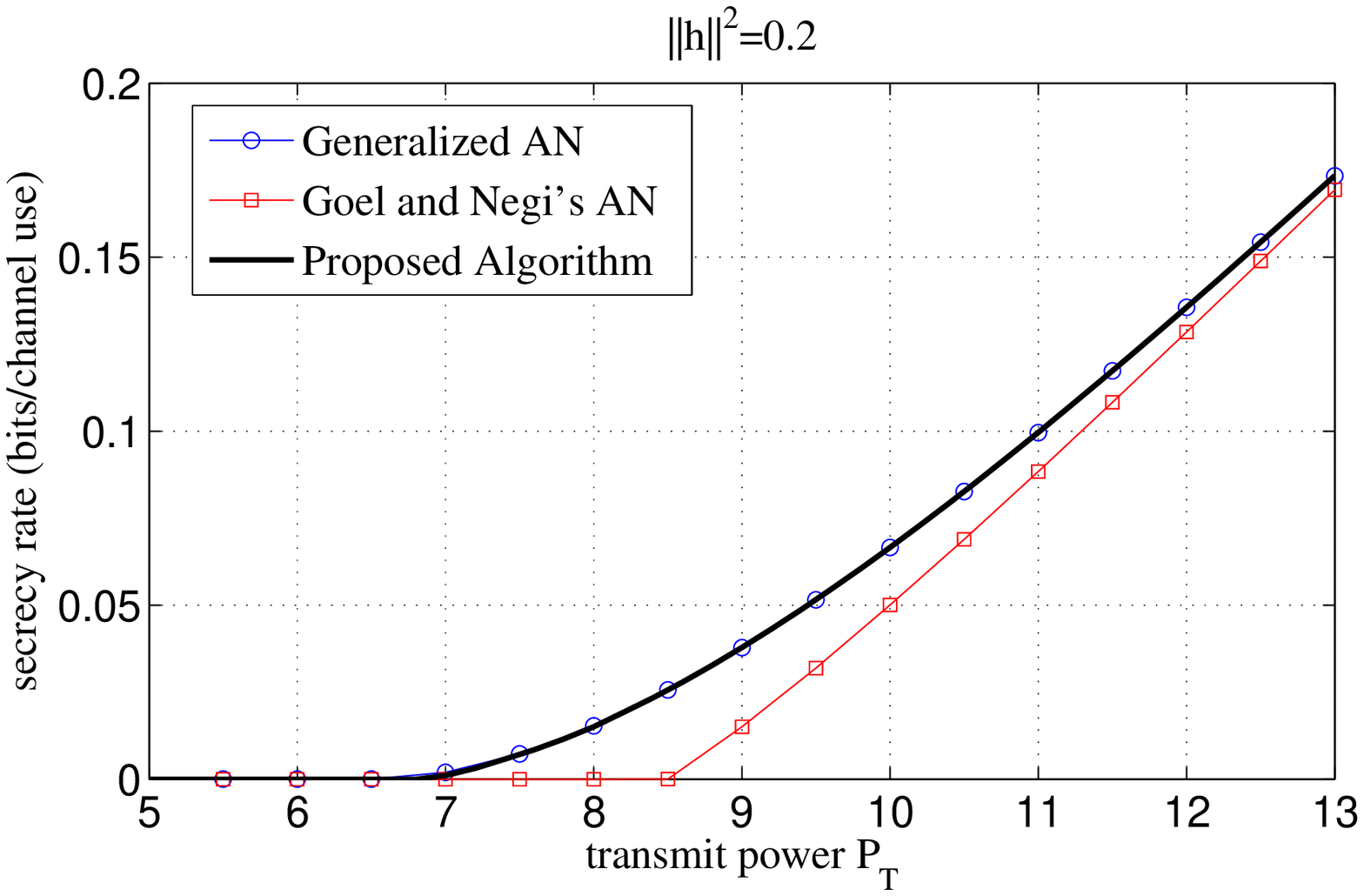 ,width=1\textwidth}
\caption{Secrecy rate versus transmit power under
$||\mathbf{h}||^2=0.2$.} \label{Fig_h_02}
\end{figure}

\begin{figure}
\centering \epsfig{file=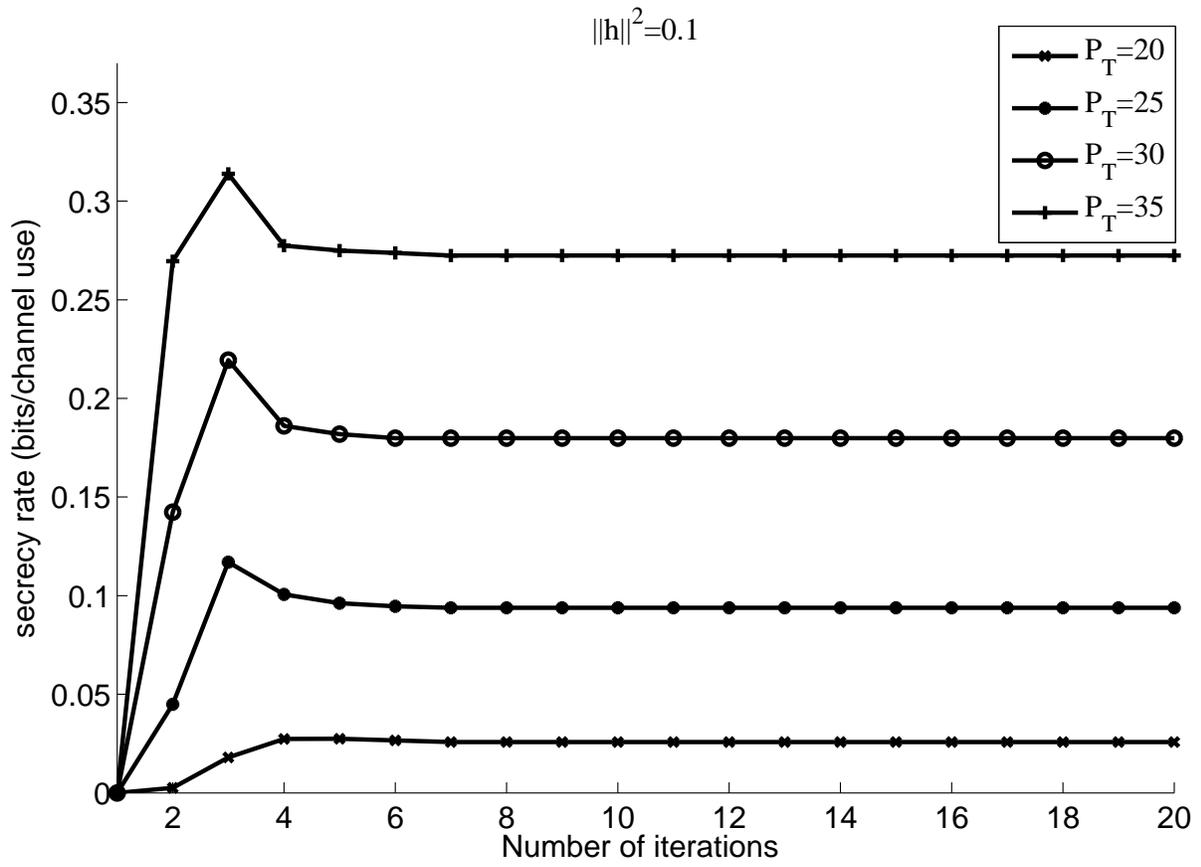 ,width=1\textwidth}
\caption{Secrecy rate versus the number of iteration under
$||\mathbf{h}||^2=0.1$.} \label{Fig_convergence01}
\end{figure}

%

\begin{figure}
\centering \epsfig{file=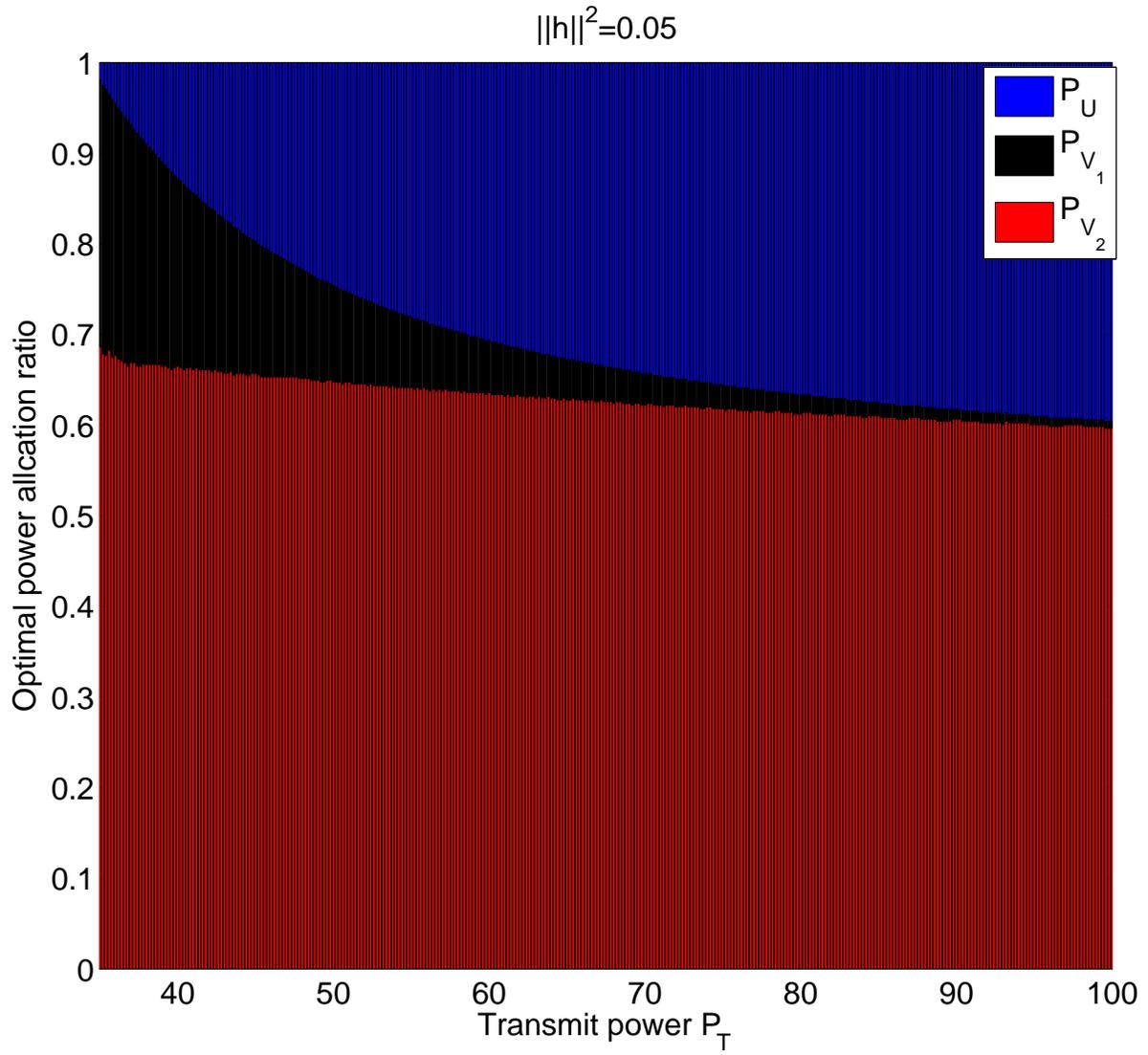 ,width=1\textwidth}
\caption{Power allocation among $P_U$, $P_{V_1}$, and $P_{V_2}$ under
$||\mathbf{h}||^2=0.05$.} \label{Fig_power_allocation_005}
\end{figure}

\end{document}